\documentclass[journal]{IEEEtran}
\usepackage{cite}
\usepackage{amsmath,amssymb,amsfonts,amsthm}
\usepackage{algorithmic}
\usepackage{graphicx}
\usepackage{textcomp}
\usepackage{subfig,booktabs,multirow}
\usepackage[unicode=true,bookmarks=false,breaklinks=false,pdfborder={0 0 1},colorlinks=false]{hyperref}
\hypersetup{	colorlinks,bookmarksopen,bookmarksnumbered,citecolor=blue,urlcolor=blue}
\usepackage [T1]{ fontenc }
\usepackage{tikz}
\usetikzlibrary{shapes,arrows,calc}
\tikzstyle{decision} = [diamond, draw, fill=blue!20, 
text width=5em, text centered, minimum width=4em, node distance=3cm, inner sep=0pt]
\tikzstyle{block1} = [draw, fill=red!30, minimum height=2em, minimum width=20em, rounded corners, text centered,text width=20em]
\tikzstyle{block2} = [draw, fill=red!30, minimum height=2em, minimum width=10em, rounded corners, text centered,text width=10em]
\tikzstyle{line} = [draw, -latex']
\tikzstyle{cloud} = [draw, ellipse,fill=red!20, node distance=3cm,
minimum height=2em]
\tikzstyle{block} = [draw, fill=red!30, minimum height=2em, minimum width=5em, rounded corners, text centered,text width=5em]
\tikzstyle{sum} = [draw, fill=red!30, circle, node distance=1cm]
\tikzstyle{input} = [coordinate]
\tikzstyle{output} = [coordinate]
\tikzstyle{pinstyle} = [pin edge={to-,thin,black}]
\tikzstyle{place} = [coordinate]
\newtheorem{theorem}{Theorem}
\newtheorem{lemma}{Lemma}
\newtheorem{remark}{Remark}
\begin{document}

	\title{Identifying Friction in a Nonlinear Chaotic System  Using a Universal Adaptive Stabilizer}
	
	\author{Ali~Wadi, 
        Shayok~Mukhopadhyay, 
        Lotfi~Romdhane.%
\thanks{A. Wadi is with the department of Mechanical Engineering, American University of Sharjah, P.O. Box 26666, Sharjah, United Arab Emirates.}
\thanks{S. Mukhopadhyay is with the department of Electrical Engineering, American University of Sharjah, P.O. Box 26666, Sharjah, United Arab Emirates.}
\thanks{L. Romdhane is with the department of Mechanical Engineering, American University of Sharjah, P.O. Box 26666, Sharjah, United Arab Emirates.}%
\thanks{Corresponding author S. Mukhopadhyay \texttt{smukhopadhyay@aus.edu}.}%
}

	
\maketitle
	
	\begin{abstract}
		This paper proposes a friction model parameter identification routine that can work with highly nonlinear and chaotic systems. The  chosen system for this study is a passively-actuated tilted Furuta pendulum, which is known to have a highly non linear and coupled model. The pendulum is tilted to ensure the existence of a stable equilibrium configuration for all its degrees of freedom, and the link weights are the only external forces applied to the system. A nonlinear analytical model of the pendulum is derived, and a continuous friction model considering static friction, dynamic friction, viscous friction, and the stribeck effect is selected from the literature. A high-gain Universal Adaptive Stabilizer (UAS) observer is designed to identify friction model parameters using joint angle measurements. The methodology is tested in simulation and validated on an experimental setup. Despite the high nonlinearity of the system, the methodology is proven to converge to the exact parameter values, in simulation, and to yield qualitative parameter magnitudes in experiments where the goodness of fit was around 85\% on average. The discrepancy between the simulation and the experimental results is attributed to the limitations of the friction model. The main advantage of the proposed method is the significant reduction in computational needs and the time required relative to conventional optimization-based identification routines. The proposed approach yielded more than 99\% reduction in the estimation time while being considerably more accurate than the optimization approach in every test performed. One more advantage is that the approach can be easily adapted to fit other models to experimental data.
	\end{abstract}
	\begin{IEEEkeywords}
		Furuta pendulum; Rotary inverted pendulum; Parameter Identification; Universal Adaptive Stabilizer; Viscous Friction; Dry Friction
	\end{IEEEkeywords}

	\section{Introduction}
	\IEEEPARstart{T}{he} Furuta pendulum (also known as a rotary inverted pendulum) is an under-actuated system, which is formed by two bodies in series. The first one is rotating around a vertical axis followed by a pendulum rotating around a horizontal axis. It was first proposed by Furuta \cite{Furuta_1992} mainly to test different control laws \cite{1999}.  The Furuta pendulum is mainly used as a testbed of nonlinear control strategies and is also of educational value. Many of the works in the literature use the Furuta pendulum to illustrate a proposed control law \cite{1999,Furuta_1992,Gafvert,Wadi18}. However, these works did not elaborate on the dynamic model of the Furuta pendulum and they limited their work to a simplified model \cite{Jadlovsk__2013,Gar_a_Alar_on_2012}.
	
	Few papers dealt with the modeling of the Furuta pendulum. They highlighted the complexity of the dynamic model and the necessity to simplify it based on several assumptions. The authors of \cite{Furuta_1992} and \cite{Cazzolato_2011} proposed a linearized model based on a small angle assumption. Others neglected one or more of the cross coupling terms relating the two rotations \cite{Furuta_1992,Gafvert,Gar_a_Alar_on_2012, Antonio_Cruz_2015,Antonio_Cruz_2014}. The authors of \cite{Cazzolato_2011} showed, through simulations, that the aforementioned assumption is not acceptable and it could have an important effect on the dynamic response. The authors of \cite{Antonio_Cruz_2015} derived a significantly simplified system and carried a brief three-second test. All the experimental works on the Furuta pendulum are directed at testing various control strategies \cite{Aguilar_2008,Freidovich_2007}. The nonlinear dynamic behavior of the system, which includes Coriolis and centrifugal forces in the two rotating frames of motion, is often overlooked or linearized. However, the knowledge of the dynamic behavior of a system is highly desirable to design advanced nonlinear controllers for a given nonlinear system. In this paper, the developed Furuta pendulum model, which encompasses the aforementioned dynamics, is validated experimentally. 
	
	The problem of identifying the different parameters of the Furuta pendulum is also treated in the literature \cite{Gar_a_Alar_on_2012, Schlee,Habib_2015}. Most of the authors dealing with this problem mentioned the high complexity of the model, especially the friction effects arising from the wide range of velocities the links go through. The parameters representing the dry and viscous frictions in the joints are one of the most difficult parameters to identify \cite{Wadi2017a,Antonio_Cruz_2015,Drewniak_2009, Makkar,Crisco_2007,Teixeira_2013, Fang_2012, Freidovich_2006}. Indeed, there is no universal model to represent damping, and the literature contains several models that attempt to capture the physics of this complex phenomenon \cite{Pennestr__2015,Keck_2017,Marques_2018}. 
	
	When the Furuta pendulum is used in a passive mode, the gravity, which is the only external load, actuates only the second joint. In this case, the Furuta pendulum does not have a stable configuration, which makes any reproducible experimentation nearly impossible. Moreover, this will cause the passive Furuta pendulum to exhibit chaotic behavior such that the motion of both links become heavily dependent and sensitive to initial conditions. This fact could explain the reason of the nonexistence in the literature of dynamic models of a passively actuated Furuta pendulum. The closest work to passive excitation was in \cite{Chenglin_Hu_2009} where friction in a cart pendulum system was identified. Nevertheless, that work decoupled the pendulum from the cart and dealt with them separately in the identification phase. A significant portion of the literature work, attempting parameter estimation, chooses an active excitation of the setup, often alongside a controller that helps stabilize the system. However, a controller might mask system dynamics and thus hinder parameter identification results. Moreover, a dynamic model of the passive system could capture the physics of the system without the perturbation of the input torque. In this paper, to actuate both joints in the passive mode, we propose to tilt the Furuta pendulum. In this case, the Furuta pendulum has a stable equilibrium configuration for both links, which is defined by gravity. The oscillations of the pendulum around this equilibrium could then be studied analytically and experimentally. After tilting, the model reveals a contribution of the gravity in the two generalized forces applied to the two joints. The existence of the stable equilibrium configuration, around which the perturbed system oscillates, is also seen. The developed model includes the effect of the static and dynamic frictions, viscous friction, and the stribeck effect. We propose a novel Universal Adaptive Stabilizer-based model identification algorithm and apply it to estimate the parameters of the friction model in the joints of the Furuta pendulum. The proposed approach is tested in simulation and in experiment, and it is compared against grey box optimization-based parameter estimation. After identification, validation is performed by simulating the system model using the identified parameters and comparing that to the experimental time response of the system. The coefficient of determination, $ R^2 $, is used as a benchmark statistic, and the computation time the algorithm takes to run is also noted for comparison with the optimization-based approach.
	
	The rest of the paper is organized as follows: after the introduction, Section II introduces the model developed. The identification approach is presented in Section III. The comparison of the results, along with a discussion, are presented in Section IV. Some concluding remarks and future work are presented in Section V.
	
	\section{Materials and Methods}
	
	\begin{figure}
		\centering
		\includegraphics[width=8.25 cm]{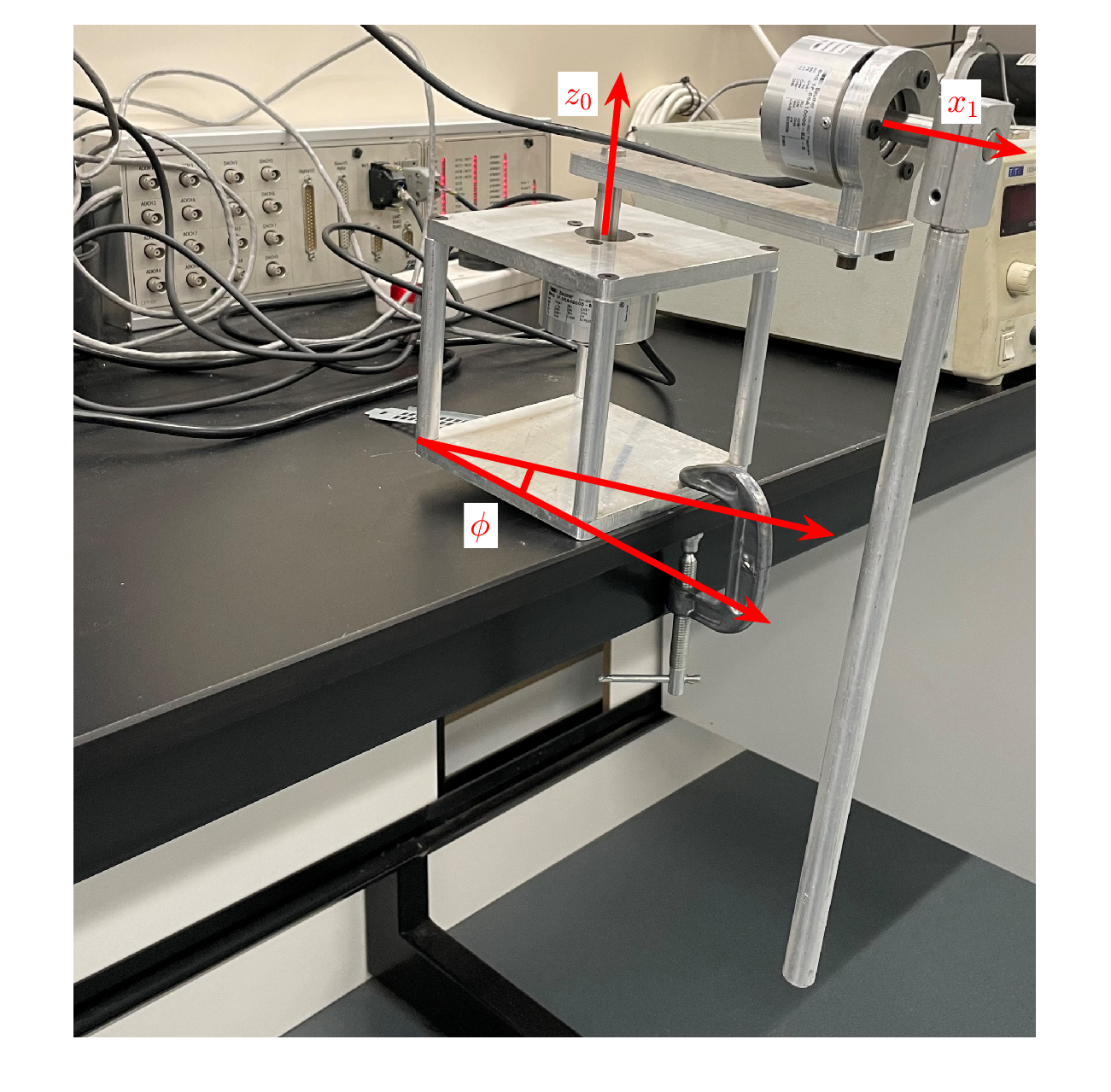}
		\caption{The Furuta pendulum in its tilted position.}\label{Hardware}
	\end{figure}  
	\begin{figure}
		\centering
		\includegraphics[width=8 cm]{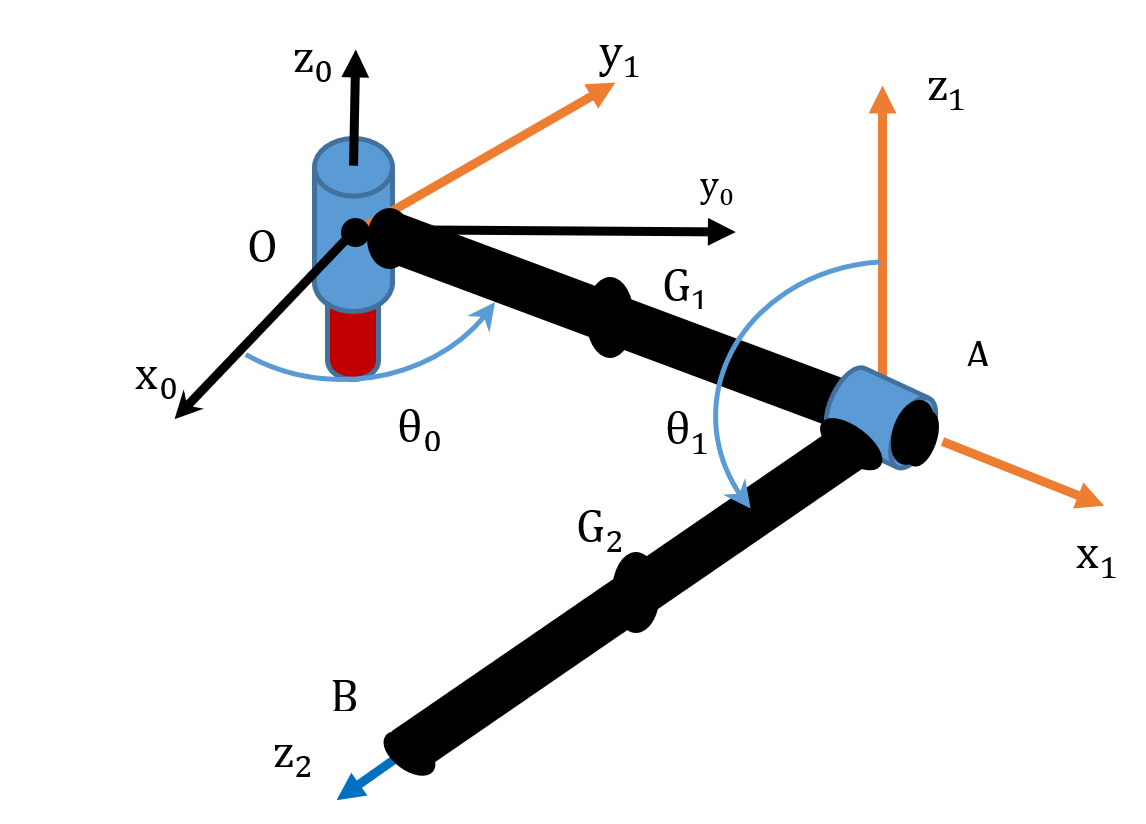}
		\caption{Schematic of the tilted setup}\label{Schematic}
	\end{figure}

		
		\begin{table*} 
			\caption{Parameters of the Furuta Pendulum Model.}
			\label{table:parameters}
			\centering
			\begin{tabular}{lcll}
				\toprule
				{Parameter}&  Unit	& {Arm link}	& {Pendulum link}\\
				\midrule
				Mass &$ [Kg] $						& $ m_1=0.370 $						& $ m_2=0.128 $\\
				Inertia &$ [Kg \, m^2] $			& $ j_{1_z} \,= 3.09\times 10^{-3} $	& $ \begin{matrix}	j_{2_{x}}=5.25\times 10^{-3} \\ j_{2_{y}}=5.25\times 10^{-3} \\ j_{2_{z}}=2.91\times 10^{-6}		\end{matrix} $\\
				Pivot to $ CG $ distance &$ [m] $	& $ l_1\;=OG_1=0.0620 $			&	$ l_2=AG_2=0.0620 $\\
				Total length &$ [m] $				& $ L_1\;=OA=0.216 $			&	$ L_2=AB=0.316 $\\
				\bottomrule
			\end{tabular}
		\end{table*}
		\subsection{The Furuta Pendulum}
		Figure~\ref{Hardware} shows the Furuta pendulum in its tilted position. The angle $ \phi $ is taken around the fixed y-axis, which is pointing out of the page. The obtained tilted reference frame is $x_0-y_0-z_0$, where $z_0$ is now the axis of the first joint between the ground and the arm. The joint is disconnected from the existing motor and the pendulum is subjected only to its weight. Figure~\ref{Schematic} shows the schematic of the Furuta pendulum in its tilted position. The arm is rotating around the fixed axis $z_0$ by an angle $\theta_0$ and the pendulum is rotating with respect to the moving axis $x_1$ by an angle $\theta_1$. The different parameters of the two links are defined in Table~\ref{table:parameters}. 
		
		\subsection{Dynamic Model of the Furuta Pendulum}
		Based on the schematic of the Furuta pendulum, shown in Figure~\ref{Schematic}, the dynamic model is elaborated. The Euler-Lagrange Formulation is used to derive the equations of the motion of the system.  Two generalized coordinates are required to fully describe the dynamics of the two-degrees-of-freedom system, and they are the angular displacements of the Arm and Pendulum links given by
		
		\begin{equation}\label{q}
			q=\begin{bmatrix}	\theta_0\\\theta_1 \end{bmatrix}
		\end{equation}
		
		The Euler-Lagrange equation, for the $ i^{th} $ generalized coordinate of the system can be written as follows:
		
		\begin{equation}\label{EL}
			\frac{d}{dt} \left( \frac{\partial T}{\partial \dot{q}_i} \right) - \frac{\partial T}{\partial q_i} + \frac{\partial V}{\partial q_i} =Q_i\,
		\end{equation}
		
		where the kinetic energy, $ T $, is written as the sum of the linear and angular kinetic energies of the two links, which  is given by
		
		\begin{equation}\label{T}
			T = \sum_{i=1}^{2} \frac{1}{2} m_i {v}_{G_i}^{2} + \sum_{i=1}^{2} \frac{1}{2} m_{i}  {\omega}_{i-1}^{T}({J}_i {\omega}_{i-1}),
		\end{equation}
		
		where ${\omega}_{0}=\dot{\theta}_0 {z}_0$, ${\omega}_{1}=\dot{\theta}_0 {z}_0+\dot{\theta}_1 {x}_1$ and $ {v}_{G_i} $ is the velocity of $G_i$ ($ i=1 $ for the Arm and $  i=2 $ for the Pendulum).
		The potential energy of gravity is obtained following the scalar product between the global vertical axis and the center of gravity position vector in the local tilted frame of reference.
		
		\begin{equation}\label{V}
			V = - m_1 g ({z} \cdot {OG_1}) - m_2 g ({z} \cdot {OG_2})
		\end{equation}
		
		The generalized forces due to friction in the joints are written in \eqref{Qf} as the sum of friction torques acting on the links. This model, proposed in \cite{Brown_2016}, is chosen for numerous reasons. It is versatile in modeling different friction phenomenon, its parameters are physically meaningful and not added on an ad hoc basis, and it is continuously differentiable. The latter two properties are required by the proposed algorithm as upper and lower limits on the estimated parameters steady-state values need to be placed in the algorithm, and the time derivative of the friction term is a necessary part in the algorithm formulation.
		
		\begin{equation}\label{Qf}
			\begin{split}
				{f}_i(\dot{\theta}_i) =& F_{n_i} \mu_{d_i} tanh \left( 4\;\frac{\dot{\theta}_i}{\dot{\theta}_{t_i}}\right) 
				+	F_{n_i}  
				\cfrac{(\mu_{s_i} -\mu_{d_i})\; \cfrac{\dot{\theta}_i}{\dot{\theta}_{t_i}} }{ \cfrac{1}{4} \left(\cfrac{\dot{\theta}_i}{\dot{\theta}_{t_i}} \right)^2 + \cfrac{3}{4} }\\
				&+ \mu_{v_i} \dot{\theta}_i tanh \left( 4\;\frac{F_{n_i}}{F_{nt_i}}\right)
			\end{split}
		\end{equation}
		
		where $ F_{n_i} $ is the normal force in joint $ i $, which is a function of the system parameters as well as the state $ q $ and its derivatives; $ \mu_{d_i} $, $ \mu_{s_i} $, and $ \mu_{v_i} $ are the dynamic, static and viscous friction coefficients; $ \dot{\theta}_{t_i} $ is the transition angular velocity responsible for shaping the stribeck curve; and $ F_{nt_i} $ is the transition force that activates the viscous friction term. 
		
		The final model is given in the following form:
		
		\begin{equation}\label{Model}
			{H(q)} {\ddot{q}} + {B(q,\dot{q})} + {G(q)}={Q(\dot{q})}
		\end{equation}
		
		where the generalized inertia matrix is given by
		
		\begin{equation}\label{H}
			{H(q)}= \begin{bmatrix}	\begin{pmatrix}j_{1_z} + m_1l_1^2 \\+ (m_2 l_2^2 + j_{2_y})\; sin(\theta_1)^2 \\+j_{2_z}\; cos(\theta_1)^2\end{pmatrix} & -m_2 l_2 L_1\; cos(\theta_1)\\ \\
				-m_2 l_2 L_1\; cos(\theta_1) & j_{2_x} + m_2 l_2^2\end{bmatrix}
		\end{equation}
		
		the centrifugal and coriolis forces are given by
		
		\begin{equation}\label{B}
			{B(q,\dot{q})}=\begin{bmatrix} \begin{pmatrix}-(m_2l_2^2-j_{2_y}+j_{2_x})\;sin(2\theta_1)\;\dot{\theta}_0\; \dot{\theta}_1\\-m_2 l_2 L_1\; sin(\theta_1)\;\dot{\theta}_1^2\end{pmatrix} \\ \\
				(m_2l_2^2-j_{2_y}+j_{2_z})sin(\theta_1)\;cos(\theta_1)\;\dot{\theta}_0^2	\end{bmatrix}
		\end{equation}
		
		the effect of gravity is given by
		
		\begin{equation}\label{G}
			{G(q)}=\begin{bmatrix} \begin{pmatrix} m_1 g l_1\; sin(\theta_0)\;sin(\phi) \\ \\ +m_2 g \begin{pmatrix} L_1\; sin(\theta_0)\\ -l_2 sin(\theta_1)\;cos(\theta_0) \end{pmatrix}\; sin(\phi) \end{pmatrix} \\ \\
				-m_2 g l_2 \begin{pmatrix}sin(\theta_0)\; cos(\theta_1)\; sin(\phi)\\ - sin(\theta_1)\; cos(\phi)\end{pmatrix}	\end{bmatrix}
		\end{equation}
		
		and the generalized forces due to friction are given by
		
		\begin{equation}\label{Q}
			{Q(\dot{q})}=\begin{bmatrix} {f}_0(\dot{\theta}_0) \\ {f}_1(\dot{\theta}_1)	\end{bmatrix}
		\end{equation}
		
		where the full term is presented in Equation~\eqref{Qf}.
		
		\subsection{Setup}
		The experimental setup is comprised of a Furuta Pendulum and a Dspace 1104 data acquisition system (Figure~\ref{Hardware}). The system is equipped with two Baumer rotary incremental encoders with a resolution of $ 40,000 $ counts/rev to record the angular positions of the two joints as a function of time. A sampling frequency of $ 10 $ \textit{kHz}  was used to capture the time response of the system. The acquisition time covers the full motion until the links stop moving.  
		
		Now that the setup is explained, the preliminaries required for Universal Adaptive Stabilizer (UAS) based parameters estimation are provided in Section \ref{uinfo}. This is followed by description of the proposed parameters identification framework in Section \ref{Sec:paraID}, and the mathematical justification in Section \ref{mjusti}.
		
		\section{UAS-based High Gain Adaptive Parameter Identification Routine}
		\subsection{Universal Adaptive Stabilizer (UAS)}\label{uinfo}
		Nussbaum functions are switching functions used in high-gain adaptive control such as with a UAS. They are employed where the control direction is generally unknown. Hence, they are selected as no assumptions and no restrictions on the parameter signs and magnitudes are made. This has been successfully employed in \cite{Ali2017,Usman2019} in estimating the parameters of various other systems. Equation~\eqref{Ni1} shows some examples of Nussbaum functions.
		
		A Nussbaum function is a piecewise right continuous and locally Lipschitz function $N(.): [k',\infty) \to \mathbb{R}$ that satisfies
		
		\begin{equation}\label{Nuss}
			\begin{aligned}
				&\sup\limits_{k>k_0} \frac{1}{k-k_0} \int_{k_0}^{k} N(\tau) \,d\tau =+ \infty\\
				&\text{and}\\
				&\inf\limits_{k>k_0} \frac{1}{k-k_0} \int_{k_0}^{k} N(\tau) \,d\tau =- \infty\\
			\end{aligned}
		\end{equation}
		A Mittag-Leffler function was chosen to act as a Nussbaum function for the purpose of tuning the design parameters. It is of interest to note that a Mittag-Leffler function can act as a Nussbaum function under certain conditions that are documented in \cite{li_when_2009}. The Mittag-Leffler function depends on two positive real parameters, $\alpha$ and $\beta$. Authors in \cite{li_when_2009} determined the conditions under which the Mittag-Leffler function acts as a Nussbaum function, which are if $\alpha\in(2,3]$ and $\beta=1$. The Nussbaum function of Mittag-Leffler type is given in \eqref{Ni1} as $N_1(-\lambda z^\alpha)$.
		
		\begin{equation}\label{Ni1}		N_1(z)=\sum_{\gamma=0}^{\infty} \frac{z^\gamma}{\Gamma(\alpha\gamma+\beta)},\, \alpha\in(2,3],\, \beta=1\end{equation}
		\begin{equation}\label{Ni2}		N_2(z)=z\; cos(\sqrt{|z|})				\end{equation}
		\begin{equation}\label{Ni3}		N_3(z)=z^{2}\; cos(|z|)					\end{equation}
		\begin{equation}\label{Ni4}		N_4(z)=cos(\frac{\pi z}{2})\;e ^{z^{2}}	\end{equation}
		
		The above functions differ in how fast they vibrate. In other works \cite{Usman2019,Ali2017}, it was advantageous to choose a rapidly vibrating Nussbaum functions in applications involving DC motor and Li-ion Battery Dynamics. For that reason, the Mittag-Leffler Nussbaum function is used in this work.
		
		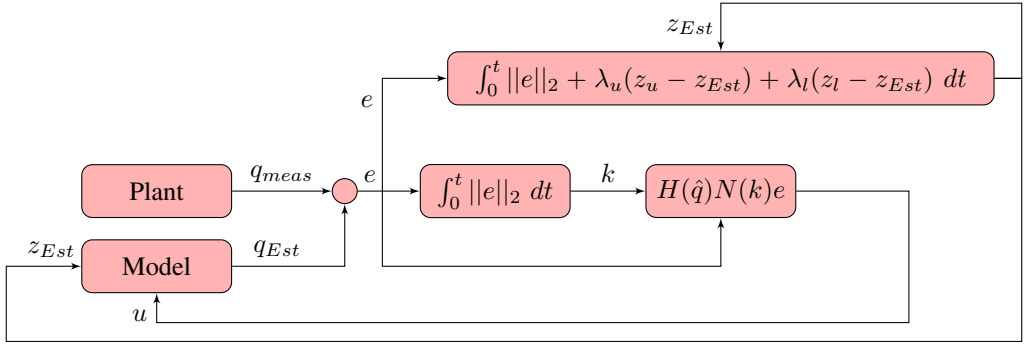
\begin{figure*}
			\centering
			\begin{tikzpicture}[auto, node distance=1.75cm,>=latex']
				\node [input] (input){};
				\node [place, below of=input, node distance=2cm] (p1) {};
				
				\node [block, right of=input, node distance=2cm] (UVP) {Plant};
				\node [block, below of=UVP, node distance=1cm] (UVM) {Model};
				\node [sum, right of=UVP, node distance=2.5cm] (sum2) {};
				\node [block, right of=sum2, node distance=2cm] (int) {$ \int_{0}^{t} ||e||_{2} \; dt$};
				\node [block, right of=int, node distance=3cm] (ML) {$H(\hat{q})N(k) e$};		
				\node [block1, above of=ML, node distance=1.5cm] (ParaUP) {$ \int_{0}^{t} ||e||_{2} + \lambda_{u}(z_{u} - z_{Est}) + \lambda_{l}(z_{l} - z_{Est}) \; dt $};
				
				\node [output, right of=ML, node distance=2.5cm] (output) {};
				\node [place, right of=sum2, node distance=.5cm] (p2) {};
				\node [place, below of=p2, node distance=1cm] (p3) {};
				\node [place, below of=UVM, node distance=0.75cm] (p4) {};
				\node [place, right of=ParaUP, node distance=4cm] (p5) {};
				\node [place, above of=p5, node distance=1cm] (p6) {};
				\node [place, above of=ParaUP, node distance=.5cm] (p7) {};
				
				\draw [->] (UVM) -| node[pos=0.2] {${q}_{Est}$} (sum2);
				\draw [->] (UVP) -- node[pos=0.5] {${q}_{meas}$} (sum2);
				\draw [->] (sum2) -- node[pos=0.2] {$e$} (int);
				\draw [->] (int) -- node {$k$} (ML);
				\draw [->] (p2) --  (p3) -| node {} (ML);
				\draw [->] (ML) -- node {} (output) |- (p4) -- node[pos=0.25] {$u$} (UVM.south);
				\draw [->] (p2) |- node[pos=0.4] {$e$} (ParaUP);
				\draw [->] (ParaUP.east) -- (p5) |- (p1) |- node[pos=0.8] {$z_{Est}$} (UVM.west);
				\draw [->] (ParaUP.east) -| (p6) -| (p7) -- node[pos=0.1,xshift=-0.8 5cm,yshift=2mm] {$z_{Est}$} (ParaUP.north);
			\end{tikzpicture}
			\caption{Algorithm Flowchart}
			\label{AlgFlow}
		\end{figure*}
		
		\subsection{Parameter Identification}\label{Sec:paraID}
		The models in Equation~\eqref{Model} and Equation~\eqref{Qf} are rewritten here to setup the parameter identification procedure as the following:
		
		\begin{equation}\label{EstModel}
			{H(\hat{q})} {\ddot{\hat{q}}} + {B(\hat{q},\dot{\hat{q}})\dot{\hat{q}}} + {G(\hat{q})}	
			=
			{\hat{Q}(\dot{\hat{q}})}	+  {u}\\
		\end{equation}
		here $ {\ddot{\hat{q}}, \dot{\hat{q}}, \hat{q}} $ are the estimated  acceleration, velocity and position, respectively, and ${u}$ is the UAS input given by Equation~\eqref{unk}. Also note that the parameters of $ H $, $ B $, and $ G $ in \eqref{EstModel} are known, and $ {\hat{Q}}=\begin{bmatrix}
			{\hat{f}_0(\dot{\hat{\theta}}_0)}&{\hat{f}_1(\dot{\hat{\theta}}_1)}
		\end{bmatrix}^T$ is the estimated friction vector whose entries are given in \eqref{EstQf1}, and \eqref{EstQf2}.
		\begin{equation}\label{EstQf1}
			\begin{split}
				\hat{f}_0(\dot{\hat{\theta}}_0) =& F_{n_0} \hat{z}_1 tanh \left( 4\;\frac{\dot{\hat{\theta}}_0}{\hat{z}_4}\right) 
				+	F_{n_0}  
				\cfrac{(\hat{z}_2 -\hat{z}_1)\; \cfrac{\dot{\hat{\theta}}_0}{\hat{z}_4} }{ \cfrac{1}{4} \left(\cfrac{\dot{\hat{\theta}}_0}{\hat{z}_4} \right)^2 + \cfrac{3}{4} }\\
				&+ \hat{z}_3 \dot{\hat{\theta}}_0 tanh \left( 4\;\frac{F_{n_0}}{\hat{z}_5}\right)
			\end{split}
		\end{equation}
		\begin{equation}\label{EstQf2}
			\begin{split}
				\hat{f}_1(\dot{\hat{\theta}}_1) =& F_{n_1} \hat{z}_6 tanh \left( 4\;\frac{\dot{\hat{\theta}}_1}{\hat{z}_9}\right) 
				+	F_{n_1}  
				\cfrac{(\hat{z}_7 -\hat{z}_6)\; \cfrac{\dot{\hat{\theta}}_1}{\hat{z}_9} }{ \cfrac{1}{4} \left(\cfrac{\dot{\hat{\theta}}_1}{\hat{z}_9} \right)^2 + \cfrac{3}{4} }\\
				&+ \hat{z}_8 \dot{\hat{\theta}}_1 tanh \left( 4\;\frac{F_{n_1}}{\hat{z}_{10}}\right)
			\end{split}
		\end{equation}
		
		In these equations $ \hat{z}_n $ represents the frictional model parameter estimates where $ z_1 = \mu_{d_0} $, $ z_2 = \mu_{s_0} $, $ z_3 = \mu_{v_0} $, $ z_4 = \dot{\theta}_{t_0} $, $ z_5 = F_{nt_0} $, $ z_6 = \mu_{d_1} $, $ z_7 = \mu_{s_1} $, $ z_8 = \mu_{v_1} $, $ z_9 = \dot{\theta}_{t_1} $, $ z_{10} = F_{nt_1} $, the following equations complete the high-gain universal adaptive observer used for parameters estimation in this work.
		
		\begin{equation}\label{e} e(t)=\dot{q}(t) -	\dot{\hat{q}}(t)  \end{equation}
		\begin{equation}\label{kd} \dot{k}(t)= \left\lVert e(t) \right\rVert^2_2,\; k(t_0)=k_0>0\end{equation}
		\begin{equation}\label{Nk} N(k(t))=N_1(-\lambda k(t)^\alpha) \end{equation}
		\begin{equation}\label{unk} u(t)=	H(\hat{q}) N(k(t))e(t) \end{equation}
		
		In Equation~\eqref{Nk}, $\lambda = 1$, and $\alpha =3$, thus resulting in a Nussbaum function of the Mittag-Leffler form. In Equation~\eqref{unk} $ H $ is the generalized inertia matrix from \eqref{H}. 
		Equation~\eqref{adapt} is used in identifying the friction model parameters $\hat{z}_n$ in Equation~\eqref{EstQf1}, and Equation~\eqref{EstQf2}.
		
		\begin{equation}\label{adapt} 
			\dot{\hat{z}}_n(t)=(\gamma+\lambda_{nu} (z_{nu}-\hat{z}_{n}(t))+\lambda_{nl} (z_{nl}-\hat{z}_{n}(t)) ) \left\lVert e(t) \right\rVert_2
		\end{equation}
		
		where $\hat{z}_n(t)\in \Bbb R$ represents each parameter that needs to be identified, and $n \in {1, 2, \dots, N}$ refers to the parameter number where $N=10$ is the number of unknown parameters that need to be identified. Also, $\gamma$ is a tuning parameter to control the rate at which the parameters adapt.
		
		The proposed parameter evolution equation requires steady state upper and lower bounds $z_{nu}$ and $z_{nl}$, as well as their confidence levels $\lambda_{nu}$ and $\lambda_{nl}$, respectively. The parameters $\lambda_{nu},\lambda_{nl}>0$ and the parameters $z_{nu},z_{nl},\lambda_{nu},\lambda_{nl} \in \Bbb R$. The upper and lower bounds $z_{nu}$ and $z_{nl}$ are constants which represent the limits within which $\hat{z}_n(t)$ is desired to settle, as $t \to \infty$. The constants $\lambda_{nu}$ and $\lambda_{nl}$ represent the users confidence in their choice of steady state upper and lower bounds $z_{nu}$ and $z_{nl}$, respectively. The flowchart of the proposed parameter identification routine is presented in Figure~\ref{AlgFlow}.
		
		\subsection{Mathematical Justification}\label{mjusti}
		
		This section presents the mathematical justification behind the proposed friction model parameters identification routine. First, Lemma~\ref{adaptiveproof} shows that the adaptive observer's error heading to zero leads to boundedness of the parameters. Second, Theorem~\ref{boundnessproof} proves that the adaptive observer's error is driven to zero by the designed UAS input. Finally, Theorem~\ref{convergeproof} establishes the convergence of the estimated parameters to the true parameters values.
		
		\begin{lemma}\label{adaptiveproof}
			Let $ \hat{z}_n $ be given as in Equation~\eqref{adapt}. Given that $ \lambda_{nu}, \lambda_{nl}>0 $, \textbf{if} the parameter adaptation proposed in Equation~\eqref{adapt} is used and 
			$ e(t)  \to 0 $ as $ t\to \infty $, \textbf{then} $ \hat{z}_n $ is bounded $ \forall t>t_0 $, for $ n \in \{1,2, \dots, 10\} $.
		\end{lemma}
		\begin{proof}[Proof of Lemma~\ref{adaptiveproof}]
			Observe that \eqref{adapt} is a stable Linear Time Invariant (LTI) system driven by the input $ (\gamma+\lambda_{nu} z_{nu}+\lambda_{nl} z_{nl} ) \left\lVert e(t) \right\rVert_2  $ as shown in Equation~\eqref{adapt_expand}.
			
			\begin{equation}\label{adapt_expand} 
				\dot{\hat{z}}_n(t) + (\lambda_{nu}+\lambda_{nl}) \hat{z}_{n}(t)=(\gamma+\lambda_{nu} z_{nu}+\lambda_{nl} z_{nl} ) \left\lVert e(t) \right\rVert_2 
			\end{equation}
			
			The solution of the above equation can be obtained by the properties of the convolution integral to be
			
			\begin{equation}\label{adapt_solution}
				\begin{split}
					&\hat{z}_n(t)=\hat{z}_n(t_0)e^{-(\lambda_{nu}+\lambda_{nl})t} \\
					&+(\gamma+\lambda_{nu} z_{nu}+\lambda_{nl} z_{nl})\int_{t_0}^{t}  \left\lVert e(t-\tau) \right\rVert_2 e^{-(\lambda_{nu}+\lambda_{nl})\tau} d\tau 
				\end{split}
			\end{equation}
			
			Because $ \hat{z}_n(t),\lambda_{nl},\lambda_{nu},z_{nl},z_{nu} $ are all bounded for $ n \in \{1,2, \dots, 10\} $, and because $ \left\lVert e(t-\tau) \right\rVert_2>0 $ for all $ t $, then Equation~\eqref{adapt_solution} yields $ \hat{z}_n(t) \to \hat{z}_{n_\infty} $ that is bounded. This can be seen from the first term on the right hand of \eqref{adapt_solution} growing smaller as $ t\to\infty $, and following the assumption that $ e\to0 $ as $ t\to\infty $, the second term under the integral is bounded. 
			
		\end{proof}
		
		With the boundedness of the parameters established in Lemma~\ref{adaptiveproof}, the assumption that the estimation error grows to zero is now proven in Theorem~\ref{boundnessproof}.

		\begin{theorem}\label{boundnessproof}
			Let $ e(t)=\dot{q}(t) -\dot{\hat{q}}(t) $ as in Equation~~\eqref{e}, where the dynamics of $ q(t) $ are given by Equation~\eqref{Qf} --- Equation~\eqref{Q}, and the dynamics of $ \hat{q}(t) $ are given by Equation~\eqref{EstModel} --- \eqref{EstQf2}, $ u $ is as given in Equation~\eqref{unk}. \textbf{If} $H(q(t))$ and $H(\hat{q}(t))$ are invertible for all $t\geq t_0$ \textbf{Then}, $ e(t) \to 0 $ as $ t\to \infty $.
		\end{theorem}
		\begin{proof}[Proof of Theorem~\ref{boundnessproof}]
			Let the system dynamics in Equation~\eqref{Model} and system model in Equation~\eqref{EstModel} be split into the dynamics and the friction counterparts and written as follows
			\begin{equation}
				\label{reDynamicModel}
				\begin{split}
					{\ddot{{q}}}&={H({q})}^{-1} [{{Q}(\dot{{q}})}- {B({q},\dot{{q}})}\dot{{q}} - {G({q})}]	\\
					{\ddot{\hat{q}}}&={H(\hat{q})}^{-1} [{\hat{Q}(\dot{\hat{q}})}- {B(\hat{q},\dot{\hat{q}})}\dot{\hat{q}} - {G(\hat{q})}]+{H(\hat{q})}^{-1}{u}	\\ 	
				\end{split}
			\end{equation}
			where it is understood that all the estimated quantities $ \hat{[\quad ]} $ vary with time, $ {u} $ is as defined in Equation~\eqref{EstModel}, and ${H({q})}^{-1}, {H(\hat{q})}^{-1}$ exist by assumption.
			
			The error dynamics $ \dot{e} ={\ddot{q}}-{\ddot{\hat{q}}} $ can then be expanded and written as
			
			\begin{equation}\label{eDynamics1}
				\begin{split}
					\dot{e}&={H}^{-1}{Q}-\hat{H}^{-1}\hat{{Q}} + \hat{H}^{-1}\hat{{B}}\dot{\hat{q}}-{H}^{-1}{B}\dot{q} \\
					&+ \hat{H}^{-1}\hat{{G}}-{H}^{-1}{G}-\hat{H}^{-1}{u}
				\end{split}
			\end{equation}
			where for ease of use we write  $H\equiv H(q), \hat H \equiv H(\hat q), G\equiv G(q), \hat G \equiv G(\hat q), B\equiv B({q},\dot{{q}}), \hat{B}=B(\hat{q},\dot{\hat{q}}), Q\equiv Q(q)$ and $\hat Q = \hat Q(\hat q)$.
			
			In Equation~\eqref{eDynamics1}, we add and subtract $ \Lambda e $, where $ \Lambda $ is a positive definite design matrix.
			
			\begin{equation}\label{eDynamics2}
				\begin{split}
					\dot{e}&={H}^{-1}{Q}-\hat{H}^{-1}\hat{{Q}} + \hat{H}^{-1}\hat{{B}}\dot{\hat{q}}-{H}^{-1}{B}\dot{q} \\
					&+ \hat{H}^{-1}\hat{{G}}-{H}^{-1}{G}-\hat{H}^{-1}{u} + \Lambda e - \Lambda e\\
					\dot{e}&={H}^{-1}{Q}-\hat{H}^{-1}\hat{{Q}} + \hat{H}^{-1}\hat{{B}}\dot{\hat{q}}-{H}^{-1}{B}\dot{q} \\
					&+ \hat{H}^{-1}\hat{{G}}-{H}^{-1}{G}-\hat{H}^{-1}{u} + \Lambda \dot{q} - \Lambda \dot{\hat{q}} - \Lambda e\\
					\dot{e}&={H}^{-1}{Q}-\hat{H}^{-1}\hat{{Q}} + \hat{H}^{-1}\hat{{G}}-{H}^{-1}{G} - \hat{H}^{-1}{u} - \Lambda e\\
					&- (\Lambda - \hat{H}^{-1}\hat{{B}})\dot{\hat{q}} + (\Lambda - {H}^{-1}{B})\dot{q}  
				\end{split}
			\end{equation}
			
			The following notation simplifications are applied to Equation~\eqref{eDynamics2}.
			\begin{equation}\label{eDynamics3}
				\begin{split}
					\dot{e}&=\underbrace{{H}^{-1}{Q}-\hat{H}^{-1}\hat{{Q}}}_\mathcal{Q} 
					+ \underbrace{\hat{H}^{-1}\hat{{G}}-{H}^{-1}{G}}_\mathcal{G}-\hat{H}^{-1}{u} - \Lambda e\\
					&- \underbrace{(\Lambda - \hat{H}^{-1}\hat{{B}})}_{\mathcal{B}_2}\dot{\hat{q}} + \underbrace{(\Lambda - {H}^{-1}{B})}_{\mathcal{B}_1}\dot{q} \\
					&=\mathcal{Q} + \mathcal{G} + \mathcal{B}_1\dot{q} - \mathcal{B}_2\dot{\hat{q}} -\hat{H}^{-1}{u} - \Lambda e
				\end{split}
			\end{equation}
			
			In Equation~\eqref{eDynamics3}, we add and subtract $ \mathcal{B}_2 \dot{q}$.
			
			\begin{equation}\label{eDynamics4}
				\begin{split}
					\dot{e}&=\mathcal{Q} + \mathcal{G} + \mathcal{B}_1\dot{q} - \mathcal{B}_2 \dot{\hat{q}} -\hat{H}^{-1}{u} - \Lambda e + \mathcal{B}_2 \dot{q} - \mathcal{B}_2 \dot{q}\\
					\dot{e}&=\mathcal{Q} + \mathcal{G} + \mathcal{B}_1\dot{q} - \mathcal{B}_2 \dot{q} + \mathcal{B}_2 (\dot{q} - \dot{\hat{q}})  -\hat{H}^{-1}{u} - \Lambda e\\
					\dot{e}&=\mathcal{Q} + \mathcal{G} + (\mathcal{B}_1-\mathcal{B}_2)\dot{q} + \mathcal{B}_2 e - \Lambda e - \hat{H}^{-1}{u} \\
					\dot{e}&=\mathcal{Q} + \mathcal{G} + (\mathcal{B}_1-\mathcal{B}_2)\dot{q} - (\Lambda - \mathcal{B}_2) e - \hat{H}^{-1}{u} \\
				\end{split}
			\end{equation}
			
			To prove the stability of the error dynamics above, pre-multiply \eqref{eDynamics4} by $ e^T $ to get
			\begin{equation}\label{eDynamics5}
				\begin{split}
					e^T\dot{e}&=e^T\mathcal{Q} + e^T\mathcal{G} + e^T(\mathcal{B}_1-\mathcal{B}_2)\dot{q} - e^T(\Lambda - \mathcal{B}_2) e \\
					&- e^T\hat{H}^{-1}{u} \\
				\end{split}
			\end{equation}
			
			Equation~\eqref{eDynamics5} can be rearranged and rewritten as the following inequality taking into account that terms like $ e^T\mathcal{Q} $ can be used to form $ 2 e^T \mathcal{Q} \leq e^Te+\mathcal{Q}^T\mathcal{Q} $ as explained in Appendix~\ref{APP}. \textbf{Also, The substitution $ u = \hat{H} N(k) e $ is used}.
			\begin{equation}\label{eDynamics6}
				\begin{split}
					e^T\dot{e} + e^T(\Lambda - \mathcal{B}_2) e&\leq 3e^Te+\mathcal{Q}^T\mathcal{Q} + \mathcal{G}^T\mathcal{G} \\
					&  + \dot{q}^T(\mathcal{B}_1-\mathcal{B}_2)^T(\mathcal{B}_1-\mathcal{B}_2)\dot{q}  \\
					&- e^T\hat{H}^{-1}\hat{H} N(k) e \\
				\end{split}
			\end{equation}
			which is further simplified considering $ \dot{k}(t)= \left\lVert e(t) \right\rVert^2_2 =e^T e  $ and that $ N(k) $ is scalar.
			\begin{equation}\label{eDynamics7}
				\begin{split}
					e^T\dot{e} + e^T(\Lambda - \mathcal{B}_2) e&\leq 3\dot{k}+\mathcal{Q}^T\mathcal{Q} + \mathcal{G}^T\mathcal{G} - N(k) \dot{k}\\
					&  + \dot{q}^T(\mathcal{B}_1-\mathcal{B}_2)^T(\mathcal{B}_1-\mathcal{B}_2)\dot{q}  \\
				\end{split}
			\end{equation}
			
			Integrating both sides of the above system yields
			\begin{equation}\label{eDynamics8}
				\begin{split}
					&\int_{t_0}^{t} e^T\dot{e} \;d\tau + \int_{t_0}^{t} e^T(\Lambda - \mathcal{B}_2) e\;d\tau\leq 
					3\int_{t_0}^{t}\dot{k}\;d\tau \\
					&+\int_{t_0}^{t}\mathcal{Q}^T\mathcal{Q}\;d\tau + \int_{t_0}^{t}\mathcal{G}^T\mathcal{G}\;d\tau - \int_{t_0}^{t}N(k) \dot{k}\;d\tau\\
					&  + \int_{t_0}^{t}\dot{q}^T(\mathcal{B}_1-\mathcal{B}_2)^T(\mathcal{B}_1-\mathcal{B}_2)\dot{q}\;d\tau  \\
				\end{split}
			\end{equation}
			Applying a change in variable in the $ \int_{t_0}^{t}N(k) \dot{k}\;d\tau $ term, we get
			\begin{equation}\label{eDynamics9}
				\begin{split}
					&\frac{1}{2}e^Te + \int_{t_0}^{t} e^T(\Lambda - \mathcal{B}_2) e\;d\tau\leq 
					3(k(t)-k(t_0)) \\
					&+\int_{t_0}^{t}\mathcal{Q}^T\mathcal{Q}\;d\tau + \int_{t_0}^{t}\mathcal{G}^T\mathcal{G}\;d\tau - \int_{k_0}^{k}N(k)\;dk\\
					&  + \int_{t_0}^{t}\dot{q}^T(\mathcal{B}_1-\mathcal{B}_2)^T(\mathcal{B}_1-\mathcal{B}_2)\dot{q}\;d\tau  \\
				\end{split}
			\end{equation}
			Dividing both sides by $ k(t)-k(t_0) $ and rearranging yields
			\begin{equation}\label{eDynamics10}
				\begin{split}
					\frac{1}{2}\frac{e^Te}{k(t)-k(t_0)} &+ \frac{\int_{t_0}^{t} e^T(\Lambda - \mathcal{B}_2) e\;d\tau}{k(t)-k(t_0)}		\leq 	3	 \\
					&-\frac{1}{k(t)-k(t_0)}\int_{k_0}^{k}N(k)\;dk\\ &+\frac{\int_{t_0}^{t}\mathcal{Q}^T\mathcal{Q}\;d\tau}{k(t)-k(t_0)} 
					+ \frac{\int_{t_0}^{t}\mathcal{G}^T\mathcal{G}\;d\tau}{k(t)-k(t_0)} \\
					&+\frac{\int_{t_0}^{t}\dot{q}^T(\mathcal{B}_1-\mathcal{B}_2)^T (\mathcal{B}_1-\mathcal{B}_2)\dot{q}\;d\tau}{k(t)-k(t_0)}  \\
				\end{split}
			\end{equation}
			Dividing both sides by the terms on the R.H.S. except for $ \frac{1}{k(t)-k(t_0)} \int_{k_0}^{k}N(k)\;dk $ yields
			\begin{equation}\label{eDynamics11}
				\begin{split}
					&\begin{pmatrix}\frac{1}{2}\frac{e^Te}{k(t)-k(t_0)} +\\ \\ \frac{\int_{t_0}^{t} e^T(\Lambda - \mathcal{B}_2) e\;d\tau}{k(t)-k(t_0)}\end{pmatrix}	\begin{pmatrix}3+\frac{\int_{t_0}^{t}\mathcal{Q}^T\mathcal{Q}\;d\tau}{k(t)-k(t_0)} 
						+ \frac{\int_{t_0}^{t}\mathcal{G}^T\mathcal{G}\;d\tau}{k(t)-k(t_0)} \\ \\
						+\frac{\int_{t_0}^{t}\dot{q}^T(\mathcal{B}_1-\mathcal{B}_2)^T (\mathcal{B}_1-\mathcal{B}_2)\dot{q}\;d\tau}{k(t)-k(t_0)}] \end{pmatrix}^{-1}	
					\leq 		 \\
					&1-\frac{\int_{k_0}^{k}N(k)\;dk}{k(t)-k(t_0)}
					\begin{pmatrix}3+\frac{\int_{t_0}^{t}\mathcal{Q}^T\mathcal{Q}\;d\tau}{k(t)-k(t_0)} 
						+ \frac{\int_{t_0}^{t}\mathcal{G}^T\mathcal{G}\;d\tau}{k(t)-k(t_0)} \\ \\
						+\frac{\int_{t_0}^{t}\dot{q}^T(\mathcal{B}_1-\mathcal{B}_2)^T (\mathcal{B}_1-\mathcal{B}_2)\dot{q}\;d\tau}{k(t)-k(t_0)}] \end{pmatrix}^{-1} \\
				\end{split}
			\end{equation}
			
			It can be noted from the above that with the exception of the $ \frac{1}{k(t)-k(t_0)} \int_{k_0}^{k}N(k)\;dk $ term on the R.H.S. of Equation~\eqref{eDynamics10}, all the terms are quadratic forms and positive. Now if the term $ k(t) $ is monotonically increasing, i.e. $ k(t) \to \infty $ as $ t \to \infty $, then the Nussbaum function term on the R.H.S. can, as a whole, take values approaching $ \pm \infty $ by the definition of a Nussbaum function as presented in Equation~\eqref{Nuss}. This violates the inequality above. Consequently, the assumption that $ k(t) \to \infty $ as $ t \to \infty $ is false and $ k(t) $ must, therefore, be bounded. However, $ \dot{k}(t) $ being a non-decreasing function by definition in \eqref{kd} and $ k(t) $ being bounded implies that $ k(t) \to k_{\infty} $ as $ t \to \infty $. This means that $ \dot{k}(t) \to 0 $ as $ t \to \infty $ i.e. $ e^{T}e \to 0 $, or $ e \to 0 $. This establishes the required result.
		\end{proof}
		
		With the estimation error going to zero established in Theorem~\ref{boundnessproof}, the convergence of the estimated parameters to the real parameters is now shown in Theorem~\ref{convergeproof}. It is also worth noting that verifying the assumptions that $H(q(t))$ and $H(\hat{q}(t))$ are invertible is simple and can be explicitly checked by numerical computation. Further, before proceeding to the next result, some notation is defined, to facilitate the a condensed representation of the next result and its proof. Thus motivated let $i\in\{0,1\}$, and $j,k\in\{1,2,\cdots,9,\ell\}$. Here we use the set $\{1,2,\cdots,9,\ell\}$ and let $\ell=10$ so that in the statement of Theorem \ref{convergeproof} the subscript value of 10 can be easily visualized. Now consider the following definitions.
		\begin{equation}\label{Paraproof21}	a_i(t)=F_{n_i} \mu_{d_i} tanh \left( 4\;\frac{\dot{\theta}_i(t)}{\dot{\theta}_{t_i}(t)}\right)\\\end{equation}
		\begin{equation}\label{Paraproof22}	\hat{a}_{ijk}(t)=F_{n_i} \hat{z}_j(t) tanh \left( 4\;\frac{\dot{{\theta}}_i(t)}{\hat{z}_k(t)}\right)\end{equation}
		\begin{equation}\label{Paraproof23}	b_i(t)=\mu_{v_i} \dot{\theta}_i(t) tanh \left( 4\;\frac{F_{n_i}}{F_{nt_i}}\right)\\\end{equation}
		\begin{equation}\label{Paraproof24}	\hat{b}_{ijk}(t)=\hat{z}_j(t) \dot{{\theta}}_i(t) tanh \left( 4\;\frac{F_{n_i}}{\hat{z}_k(t)}\right)\\\end{equation}
		\begin{equation}\label{Paraproof25}	c_i(t)=\cfrac{F_{n_i}  \mu_{s_i}\; \cfrac{\dot{\theta}_i(t)}{\dot{\theta}_{t_i}(t)} }{ \cfrac{1}{4} \left(\cfrac{\dot{\theta}_i(t)}{\dot{\theta}_{t_i}(t)} \right)^2 + \cfrac{3}{4} }\\\end{equation}
		\begin{equation}\label{Paraproof26}	\hat{c}_{ijk}(t)=\cfrac{F_{n_i} \hat{z}_j(t)\; \cfrac{\dot{{\theta}}_i(t)}{\hat{z}_k(t)} }{ \cfrac{1}{4} \left(\cfrac{\dot{{\theta}}_i(t)}{\hat{z}_k(t)} \right)^2 + \cfrac{3}{4} } \\\end{equation}
		\begin{equation}\label{Paraproof28}	d_i(t)=	\cfrac{F_{n_i}  \mu_{d_i}\; \cfrac{\dot{\theta}_i(t)}{\dot{\theta}_{t_i}(t)} }{ \cfrac{1}{4} \left(\cfrac{\dot{\theta}_i(t)}{\dot{\theta}_{t_i}(t)} \right)^2 + \cfrac{3}{4} }\\\end{equation}
		\begin{equation}\label{Paraproof27}	\hat{d}_{ijk}(t)=\cfrac{F_{n_i}  \hat{z}_j(t)\; \cfrac{\dot{{\theta}}_i(t)}{\hat{z}_k(t)} }{ \cfrac{1}{4} \left(\cfrac{\dot{{\theta}}_i(t)}{\hat{z}_k(t)} \right)^2 + \cfrac{3}{4} }\end{equation}
		\begin{theorem}\label{convergeproof}
			Considering the definitions in \eqref{Paraproof21}-\eqref{Paraproof27} let $$A(t)=\begin{bmatrix}
				{a_0-\hat{a}_{014}}&{b_0-\hat{b}_{035}}&{c_0-\hat{c}_{024}}&{d_0-\hat{d}_{014}}\\
				{a_1-\hat{a}_{169}}&{b_1-\hat{b}_{18\ell}}&{c_1-\hat{c}_{179}}&{d_1-\hat{d}_{169}}
			\end{bmatrix},$$ and $x=\begin{bmatrix}
				{1}&{1}&{1}&{1}
			\end{bmatrix}^T$. For the parameter identification problem described in Section~\ref{Sec:paraID}, suppose the conditions required for Theorem 1 to hold are satisfied. Suppose there exists a time instant $t^*>t_0$ such that for $t\to\infty, t>t^*$, $x$ does not belong to the nullspace of $A(t), A(t)\neq 0$,and all the following functions: $tanh \left( 4\;\frac{\dot{\theta}_i(t)}{\dot{\theta}_{t_i}(t)}\right)$, $tanh \left( 4\;\frac{\dot{{\theta}}_i(t)}{\hat{z}_k(t)}\right)$, $tanh \left( 4\;\frac{F_{n_i}}{F_{nt_i}}\right)$, $tanh \left( 4\;\frac{F_{n_i}}{\hat{z}_k(t)}\right)$ tend to 1. \textbf{Then} the parameter adaptation law in Equation~\eqref{adapt} leads to $ \hat{z}_n \to z_{n_{true}} $ as $ t\to \infty $ for $ n \in \{1,2, \dots, 10\} $.
		\end{theorem}
		
		\begin{proof}[Proof of Theorem~\ref{convergeproof}]
			By the assumptions Theorem~\ref{boundnessproof} is satisfied, this gives $ e\to 0 $ as $ t\to \infty $. As a result $ k(t)\to k_\infty $ (a constant), which further leads $ {u}\to 0 $ following the definition of $u$ in \eqref{Nk}-\eqref{unk}. As $ e\to 0 $, therefore by definition of $e$ in \eqref{e} we get that ${\dot{\hat{q}}} \to {\dot{q}}$, and $ {\ddot{\hat{q}}} \to {\ddot{q}} $. 
			
			Considering this, the model of the system in Equation~\eqref{Model} and the estimator dynamics in Equation~\eqref{EstModel} can be written as
			\begin{equation}\label{Paraproof}
				\begin{split}
					{H(q)}{\ddot{{q}}}+ {B({q},\dot{{q}})}\dot{q} + {G({q})}&={{Q}(\dot{{q}})}	\\
					{H(\hat{q})} {\ddot{\hat{q}}} + {B(\hat{q},\dot{\hat{q}})\dot{\hat{q}}} + {G(\hat{q})}	&=	{\hat{Q}(\dot{\hat{q}})}	+  {u}.\\
				\end{split}
			\end{equation} 
			Because $ u \to 0 $, $ {\dot{\hat{q}}} \to {\dot{q}} $, $ {\ddot{\hat{q}}} \to {\ddot{q}} $, subtracting the lower equation in \eqref{Paraproof} from the upper one, and substituting $ \hat{q} = q $ provides 
			\begin{equation}\label{Paraproof0}
				{Q}(\dot{{q}}) - \hat{Q}(\dot{{q}})	= \begin{bmatrix}	{f}_0(\dot{{\theta}}_0) \\ {f}_1(\dot{{\theta}}_1) \end{bmatrix} - 
				\begin{bmatrix}	\hat{f}_0(\dot{{\theta}}_0) \\ \hat{f}_1(\dot{{\theta}}_1) \end{bmatrix} =0.
			\end{equation}
			Now using the definitions of the friction force model in \eqref{Qf}, and the concerned estimates as defined in \eqref{EstQf1}, \eqref{EstQf2} one can rewrite \eqref{Paraproof0} as \eqref{Paraproof1} using the notation in \eqref{Paraproof21}-\eqref{Paraproof27}.
			\begin{equation}\label{Paraproof1}
				\underbrace{\begin{bmatrix}
						{a_0-\hat{a}_{014}}&{b_0-\hat{b}_{035}}&{c_0-\hat{c}_{024}}&{d_0-\hat{d}_{014}}\\
						{a_1-\hat{a}_{169}}&{b_1-\hat{b}_{18\ell}}&{c_1-\hat{c}_{179}}&{d_1-\hat{d}_{169}}
				\end{bmatrix}}_{A(t)}	\underbrace{\begin{bmatrix} 1 \\1\\1 \\ 1 \end{bmatrix}}_{x}=0
			\end{equation}
			We have arrived at \eqref{Paraproof1} because as $t\to\infty$ we have $e\to 0$. Recall that, by the assumptions of the theorem, there exists a time instant $t^*>t_0$ such that for $t\to\infty, t>t^*$, $x$ does not belong to the nullspace of $A(t), A(t)\neq 0$.  So for such $t>t^*$ from \eqref{Paraproof1} we can write ${\hat{a}_{014}=a_0}$, ${\hat{b}_{035}=b_0}$, ${\hat{c}_{024}=c_0}$, ${\hat{d}_{014}=d_0}$,
			${\hat{a}_{169}=a_1}$, ${\hat{b}_{18\ell}=b_1}$, ${\hat{c}_{179}=c_1}$, ${d_1=\hat{d}_{169}}$.
			Now considering ${\hat{a}_{014}=a_0}$, ${\hat{b}_{035}=b_0}$, ${\hat{c}_{024}=c_0}$, ${\hat{d}_{014}=d_0}$ along with their definitions in \eqref{Paraproof21}-\eqref{Paraproof27}, the following are obtained.
			\begin{equation}\label{Param1}
				F_{n_0} \hat{z}_1 tanh \left( 4\;\frac{\dot{{\theta}}_0}{\hat{z}_4}\right) 
				=
				F_{n_0} \mu_{d_0} tanh \left( 4\;\frac{\dot{\theta}_0}{\dot{\theta}_{t_0}}\right)
			\end{equation}
			\begin{equation}\label{Param2}
				\hat{z}_3 \dot{{\theta}}_0 tanh \left( 4\;\frac{F_{n_0}}{\hat{z}_5}\right)
				=
				\mu_{v_0} \dot{\theta}_0 tanh \left( 4\;\frac{F_{n_0}}{F_{nt_0}}\right)
			\end{equation}
			\begin{equation}\label{Param3}
				\cfrac{F_{n_0} \hat{z}_2\; \cfrac{\dot{{\theta}}_0}{\hat{z}_4} }{ \cfrac{1}{4} \left(\cfrac{\dot{{\theta}}_0}{\hat{z}_4} \right)^2 + \cfrac{3}{4} } 
				=
				\cfrac{F_{n_0}  \mu_{s_0}\; \cfrac{\dot{\theta}_0}{\dot{\theta}_{t_0}} }{ \cfrac{1}{4} \left(\cfrac{\dot{\theta}_0}{\dot{\theta}_{t_0}} \right)^2 + \cfrac{3}{4} }\
			\end{equation}
			
			
			As per the assumptions of the theorem, all the $tanh(\cdot)$ terms tend to unity. Considering the equations \eqref{Param1} and Equation~\eqref{Param2} with this gives $\hat{z}_1 	=	\mu_{d_0} $, $\hat{z}_3 	=	\mu_{v_0}$.
			Now considering \eqref{Param1} again with $\hat{z}_1 	=	\mu_{d_0}$ and explcitly writing the $tanh(\cdot)$ terms gives
			\begin{equation}\label{tanz4}
				tanh \left( 4\;\frac{\dot{{\theta}}_0}{\hat{z}_4}\right) = tanh \left( 4\;\frac{\dot{\theta}_0}{\dot{\theta}_{t_0}}\right).
			\end{equation}
			Now from the definition of the hyperbolic tangent function, and from its graph, it is known that its inverse can be found, i.e. so let $p=tanh \left( 4\;\frac{\dot{{\theta}}_0}{\hat{z}_4}\right) = tanh \left( 4\;\frac{\dot{\theta}_0}{\dot{\theta}_{t_0}}\right)$. Then, $tanh^{-1}(p)=  4\;\frac{\dot{{\theta}}_0}{\hat{z}_4} =   4\;\frac{\dot{\theta}_0}{\dot{\theta}_{t_0}}$, leading to $ \hat{z}_4 = \dot{\theta}_{t_0} $. Now using this exact same process with \eqref{Param2}, considering $\hat{z}_3 	=	\mu_{v_0}$ and inverting the $\tanh(\cdot)$ terms provides $  \hat{z}_5 = F_{nt_0} $. Further using the fact that $\hat{z}_4 = \dot{\theta}_{t_0} $ and considering \eqref{Param3} gives us $\hat{z}_2	=	\mu_{s_0}$.
			
			
			
			
			The exact same procedure as followed from after \eqref{Paraproof1} to the line above, can be used to explicitly analyze ${\hat{a}_{169}=a_1}$, ${\hat{b}_{18\ell}=b_1}$, ${\hat{c}_{179}=c_1}$, ${d_1=\hat{d}_{169}}$. Using this process provides the  following to $ z_6 = \mu_{d_1} $, $ z_7 = \mu_{s_1} $, $ z_8 = \mu_{v_1} $, $ z_9 = \dot{\theta}_{t_1} $, $ z_{10} = F_{nt_1} $. This completes the proof.
		\end{proof}
		It is worth mentioning that the assumptions of the above theorem are not necessarily restrictive. This is because during experimentation it has been routinely observed that all the $tanh(\cdot)$ terms frequently approach unity. So it is easy to pick time intervals at which $e\approx 0$, and $tanh(\cdot)\approx 1$ so that the the values of $\hat{z}_i(t),i\in\{1,2,\cdots,10\}$ in such an interval of time can be used as candidates for parameters estimates. Also, once such estimated parameters values are available, they can be fed into \eqref{Paraproof21}-\eqref{Paraproof27} to reconstruct matrix $A$ mentioned in Theorem \ref{convergeproof}. This will easily allow ascertaining if the nullspace related assumptions in Theorem \ref{convergeproof} are met. If they are not met then estimates from another interval in time where $e\approx 0$, and $tanh(\cdot)\approx 1$, can be checked for satisfaction of the nullspace assumption.
		
		\begin{table*}
			\centering
			\caption{Friction model parameter identification problem setup.}
			\label{table:ID0}
			\begin{tabular}{lccccc}
				\toprule
				{Parameter}&{$\mu_{d_0}$}	& {$\mu_{s_0}$}	& {$\mu_{v_0}$}  & {$\dot{\theta}_{t_0}$} &  {$ F_{nt_0} $} \\
				\midrule
				$ Z_{nl} $&$2.22\times 10^{-16}$&$2.22\times 10^{-16}$&$2.22\times 10^{-16}$&$2.22\times 10^{-16}$&$2.22\times 10^{-16}$\\
				\midrule
				$ Z_{nu} $ &$0.0750$&$0.0750$&$0.0100$&$0.0100$&$0.1$\\
				\midrule
				$ \lambda_l $	&$50$&$50$&$50$&$50$&$50$\\
				\midrule
				$ \lambda_u $	&$1$&$1$&$1$&$1$&$1$\\
				\midrule[1pt]
				{Parameter}&{$\mu_{d_1}$}	& {$\mu_{s_1}$}	& {$\mu_{v_1}$}  & {$\dot{\theta}_{t_1}$} &  {$ F_{nt_1} $}\\
				\midrule
				$ Z_{nl} $&$2.22\times 10^{-16}$&$2.22\times 10^{-16}$&$2.22\times 10^{-16}$&$2.22\times 10^{-16}$&$2.22\times 10{-16}$\\
				\midrule
				$ Z_{nu} $ &$0.150$&$0.151$&$0.0100$&$0.0100$&$0.100$\\
				\midrule
				$ \lambda_l $	&$50$&$50$&$50$&$50$&$50$\\
				\midrule
				$ \lambda_u $	&$1$&$1$&$1$&$1$&$1$\\
				\bottomrule
			\end{tabular}
		\end{table*}
		
		\section{Results and Discussion}
		
		In this section, the proposed parameter identification algorithm is tested in simulation, then experimental validation is performed. In all experiments, the links start with a zero initial angular velocity and a specified initial angular position. The motion is monitored until the system returns to the equilibrium position. 
		
		\subsection{Simulation}
		Here, the performance of the proposed algorithm is evaluated in simulation. The dynamic equations of motion for the Furuta pendulum were simulated from some initial conditions, and the angular position and velocity states are fed into the algorithm. The initial conditions were $ \theta_0(0) = 0^\circ $, $ \theta_1(0) = 120^\circ $, $ \dot{\theta}_0(0) = 0^\circ/sec $, and $ \dot{\theta}_1(0) = 0^\circ/sec $. The Runge-Kutta method of order $ 4 $ with a sampling time of $ 0.001 s $ was used to propagate the equations in time and solve the ODEs. Normally distributed Gaussian random white noise characterized by $ v \sim \mathcal{N}(0,\,\sigma=0.1)\, $ was injected into the simulated states to serve as measurements noise. Results achieved despite this, show the resilience of the proposed approach to noise. The initial conditions fed into the algorithm were also injected with similar noise to account for the difficulty associated with acquiring highly accurate and precise initial conditions. Additionally, this serves to simulate the ramifications of initial conditions mismatch in chaotic systems such as the Furuta pendulum. The adaptation parameters used to test the algorithm are detailed in Table~\ref{table:ID0}. The upper/lower parameters limits in Table~\ref{table:ID0} were used to form a normal distribution centered around the mean of each respective parameter bounds. The initial friction model parameters were sampled from this distribution to start the algorithm. The actual parameters and the initial guesses used to simulate the response are presented in Table~\ref{table:IDRSim}.
		
		\begin{figure*}
			\centering
			\includegraphics[width=16 cm]{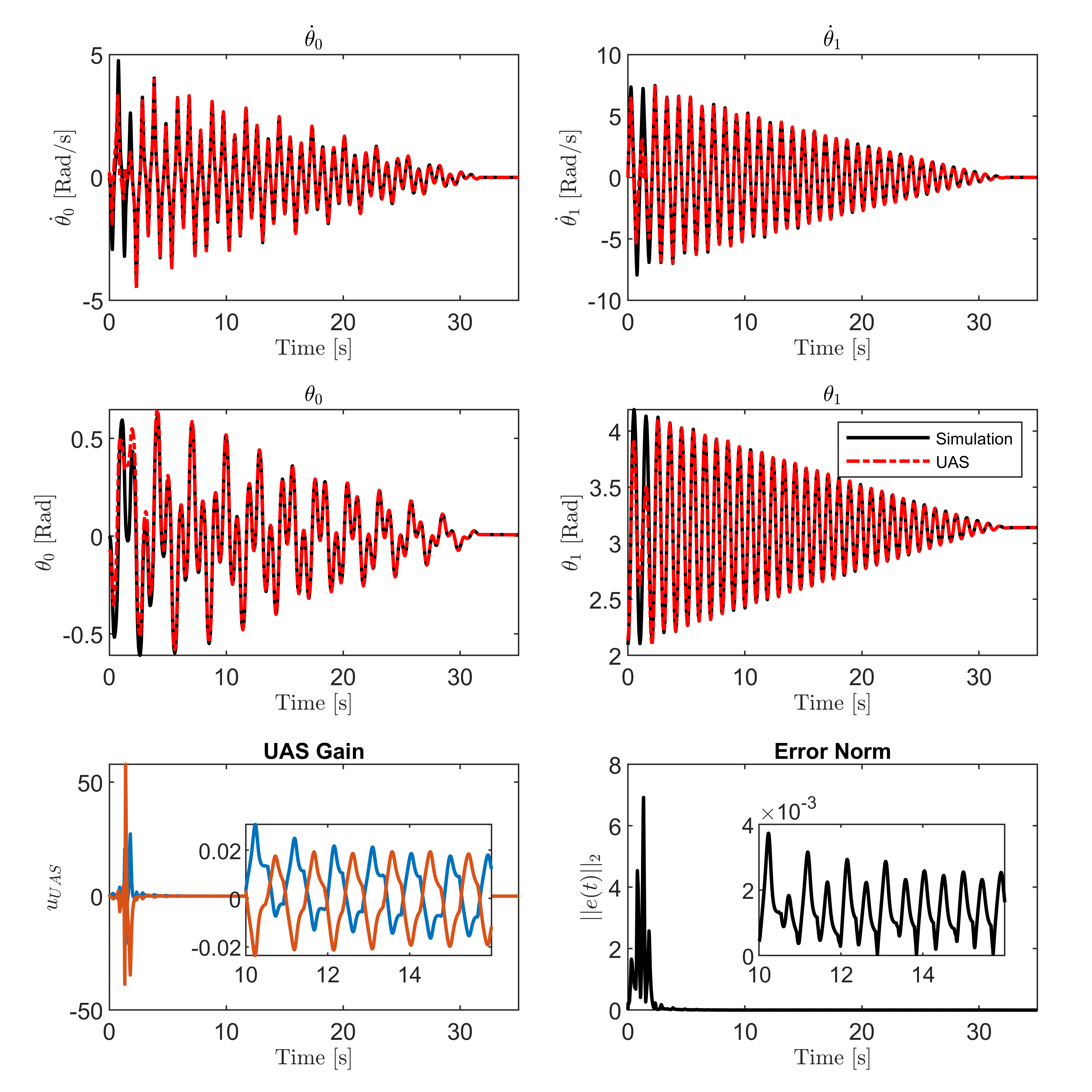}
			\caption{Performance of the UAS observer during the simulation.\label{Sim_UAS}}
		\end{figure*}
		\begin{figure*}
			\centering
			\includegraphics[width=16 cm]{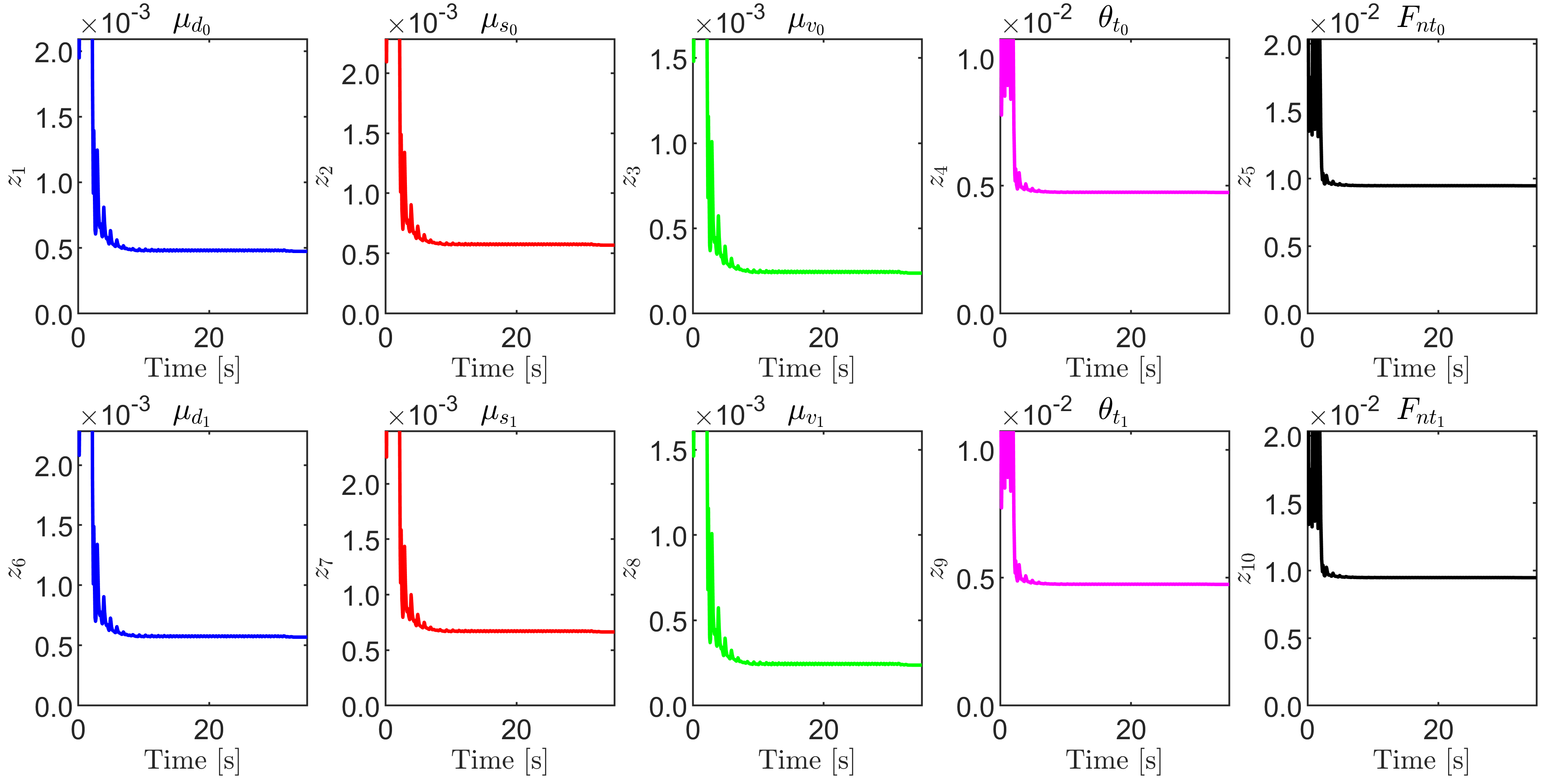}
			\caption{Parameter trajectories during the simulation.\label{Sim_Para}}
		\end{figure*}
		
		\begin{figure*}
			\centering
			\includegraphics[width=14 cm]{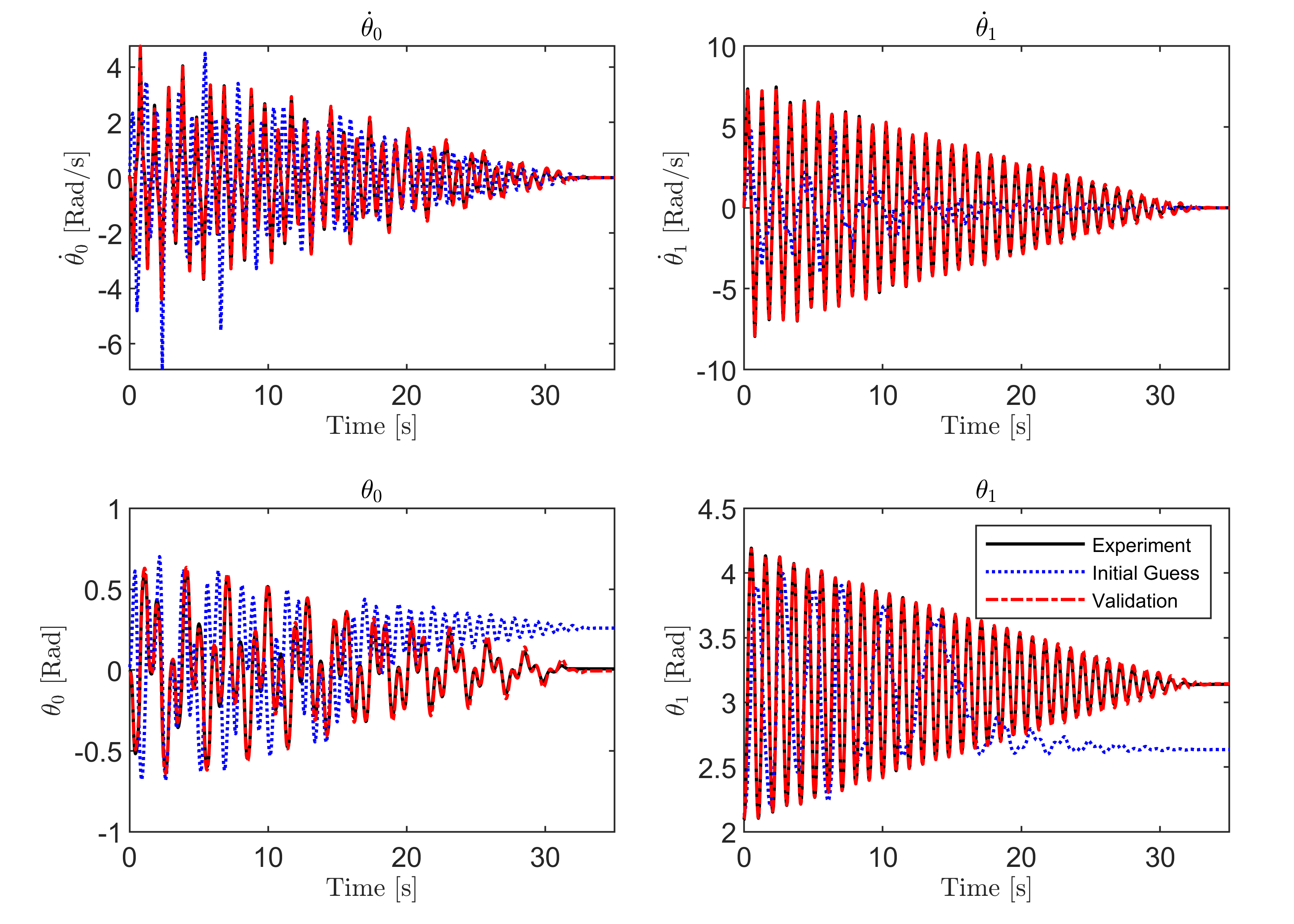}
			\caption{Validation of the identified parameters during the simulation.\label{Sim_Val}}
		\end{figure*} 
		
		\begin{figure*}
			\centering
			\includegraphics[width=14 cm]{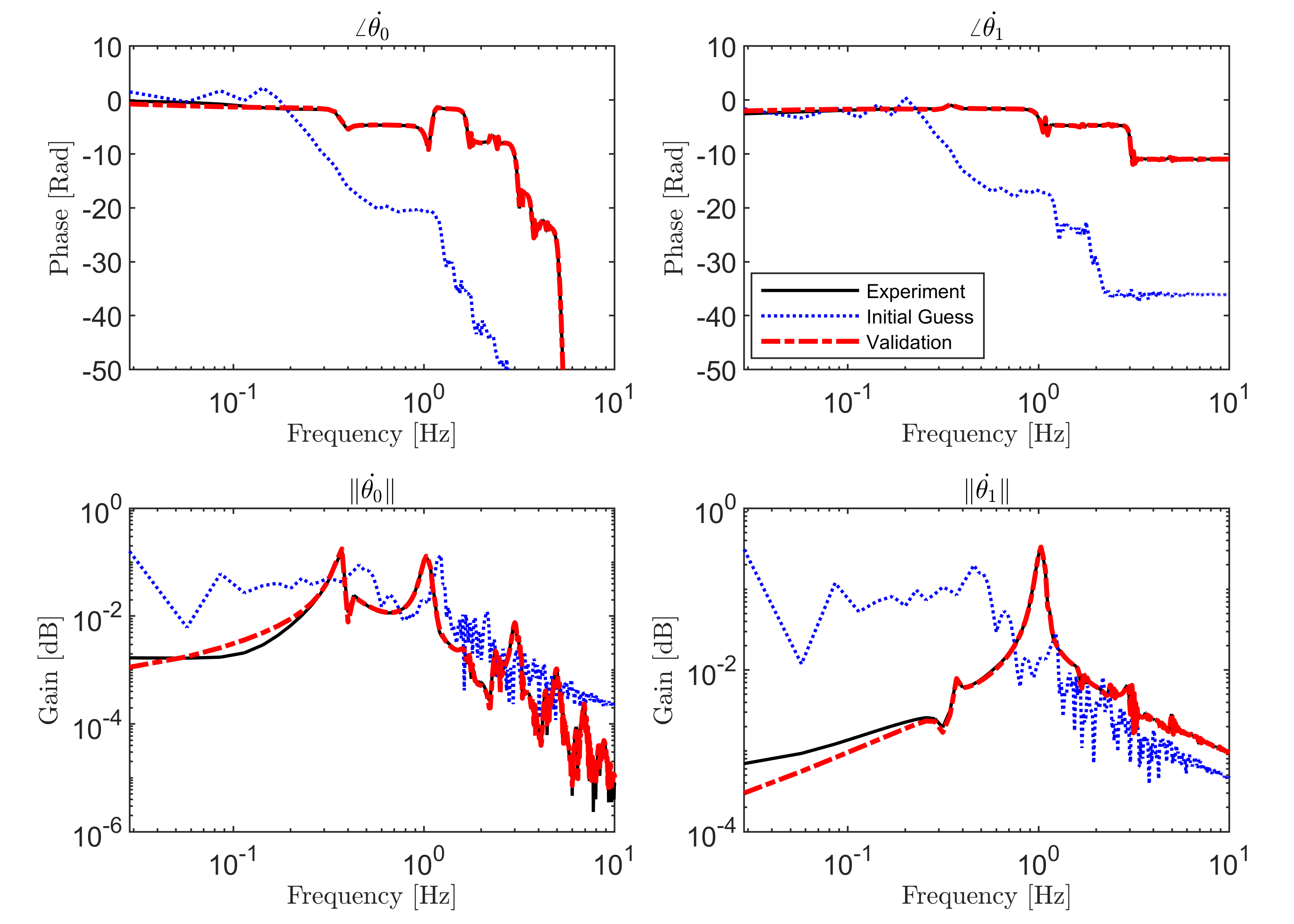}
			\caption{Frequency portrait of the identified parameters from the simulation.\label{Sim_Val_Frq}}
		\end{figure*} 
		\begin{table*}
			\centering
			\caption{Friction model parameter simulation identification results with the UAS method.\label{table:IDRSim}}
			\begin{tabular}{lccccc}
				\toprule
				{Parameter}&{$\mu_{d_0}$}	& {$\mu_{s_0}$}	& {$\mu_{v_0}$}  & {$\dot{\theta}_{t_0}$} &  {$ F_{nt_0} $} \\
				\midrule
				{Initial Guess}&$5.135\times 10^{-3}$&$5.875\times 10^{-3}$&$2.531\times 10^{-3}$&$4.853\times 10^{-2}$&$1.029\times 10^{-1}$\\
				\midrule
				{Estimate}&$4.793\times 10^{-4}$&$5.740\times 10^{-4}$&$2.424\times 10^{-4}$&$4.742\times 10^{-3}$&$9.479\times 10^{-3}$\\
				\midrule
				{Actual Value}&$5\times 10^{-4}$&$6\times 10^{-4}$&$2.5\times 10^{-4}$&$5\times 10^{-3}$&$10\times 10^{-3}$\\
				\midrule
				{Parameter}&{$\mu_{d_1}$}	& {$\mu_{s_1}$}	& {$\mu_{v_1}$}  & {$\dot{\theta}_{t_1}$} &  {$ F_{nt_1} $} \\
				\midrule
				{Initial Guess}&$5.705\times 10^{-3}$&$6.683\times 10^{-3}$&$2.399\times 10^{-3}$&$4.820\times 10^{-2}$&$9.720\times 10^{-2}$\\
				\midrule
				{Estimate}&$5.740\times 10^{-4}$&$6.688\times 10^{-4}$&$2.424\times 10^{-4}$&$4.742\times 10^{-3}$&$9.479\times 10^{-3}$\\
				\midrule
				{Actual Value}&$6\times 10^{-4}$&$7\times 10^{-4}$&$2.5\times 10^{-4}$&$5\times 10^{-3}$&$10\times 10^{-3}$\\
				\bottomrule
			\end{tabular}
		\end{table*}
		Figure~\ref{Sim_UAS} shows the UAS tracking results where it is noticed that the estimated states quickly approach and converge to the true simulated response of the system. Further, as the estimated states approach the simulated states, the norm of the error is observed to subside in magnitude and approaches zero, and the UAS gain is observed to also grow smaller as a consequence of the reduction in estimation error. Figure~\ref{Sim_Para} shows the time trajectories of the parameters of the friction model as the UAS converges to the correct states of the real system. It is seen that the parameters do indeed stabilize around some final values. The values considered to be the parameter estimates were taken just after the error norm reduced below some user defined threshold, which here is recommended to be below $ \left\lVert e(t) \right\rVert_2=0.01 $. The identified values are presented in Table~\ref{table:IDRSim}. Figure~\ref{Sim_Val} showcases a validation test of the chosen parameter estimates. The experimental response is plotted against the response realized from both the UAS estimated parameters as well as the initial parameters. The UAS-provided parameter estimates generate a response very similar to the simulated response, which establishes that the proposed approach was successful at identifying high quality parameter estimates in simulation.
		
		Further analysis is presented in frequency domain in Figure~\ref{Sim_Val_Frq} where both the phase and spectrum of the time responses of the angular positions of the arm and pendulum links are presented. It is evident that the algorithm is successful at moving from the initial estimates to ones that are more representative of the system. The portrait is quite complex considering the nonlinearity of the system, and it helps explain the compression and expansion effect present in the time responses presented as a wide array of harmonics is present in the response. The spectrum plots show good matching between estimated and actual parameters in this test.
		
		\subsection{Experiment}
		Here, the performance of the proposed algorithm is validated experimentally. The initial conditions were approximately $ \theta_0(0) = 0.2^\circ $, $ \theta_1(0) = 126^\circ $, $ \dot{\theta}_0(0) = 0^\circ/sec $, and $ \dot{\theta}_1(0) = 0^\circ/sec $. As with the test in simulation, the initial friction model parameter estimates were randomly generated also considering the upper/lower limits from Table~\ref{table:ID0}. 
		The Runge-Kutta method of order $ 4 $ with a sampling time of $ 0.0001 s $ was also used here to propagate the equations in time and solve the ODEs. The adaptation parameters used to test the algorithm are detailed in Table~\ref{table:ID0}.
		
		Figure~\ref{Exp_UAS} shows the UAS tracking results where it is noticed that the estimated states quickly approach and converge to the true response of the system. Throughout that process, the norm of the error is observed to subside in magnitude and approach zero, and the UAS gain is observed to follow a similar trend to the estimation error norm. Unlike the simulated case in the previous section, the UAS gain does not completely reduce to zero and intermittent spikes are observed in the response. Figure~\ref{Exp_Para} shows the time trajectories of the parameters of the friction model as the UAS tracks the states of the real system. It is seen that the parameters settle and oscillate within some final range of values, which makes sense following the behavior of the UAS gain. The values considered to be the parameter estimates were taken just after the error norm reduced below some user defined threshold, which here is recommended to be below $ \left\lVert e(t) \right\rVert_2=0.05 $. Due to the extra uncertainty present in the experimental data, it is recommended to average the parameters estimates during adaptation between the error subsiding and the system motion seizing. This provided tighter bounds on the parameters estimates in the experimental test. The identified values are presented in Table~\ref{table:IDRExp}. Figure~\ref{Exp_Val} showcases a validation test of the chosen parameter estimates. The experimental response is plotted against the response realized from both the UAS estimated parameters as well as the initial parameter guesses. It is clear that the UAS-provided parameter estimates generate a response that is qualitatively similar to the experiment. However, unlike the simulation test case, an exact match is not realized between the experimental response and the estimated one. We reason this to be due to the limitations of friction models available in the literature in that they only approximate and not fully characterize the damping phenomenon. Nevertheless, the proposed approach proved to be computationally highly efficient at identifying quality parameter estimates using experimental data.
		
		The same analysis is performed  in the frequency domain. In Figure~\ref{Exp_Val_Frq}, the phase and the spectrum of the time responses of the angular positions of the arm and pendulum links are presented. Similar to the experimental results, the algorithm is successful at moving from the initial estimates to the ones that are more representative of the system. However, the spectrum plots do not show as good matching between the estimated and the actual parameters in this test as it is the case with simulation. This is due to the difficulty concerning the description of the actual friction present in the system. Still, very good qualitative parameter estimates were realized using the proposed approach. In the next section, we explore alternative methods to address this problem. 
		\begin{table*}
			\centering
			\caption{Friction model parameter experimental identification results with the UAS method.\label{table:IDRExp}}
			\begin{tabular}{lccccc}
				\toprule
				{Parameter}&{$\mu_{d_0}$}	& {$\mu_{s_0}$}	& {$\mu_{v_0}$}  & {$\dot{\theta}_{t_0}$} &  {$ F_{nt_0} $} \\
				\midrule
				{Initial Guess}&$1.090\times 10^{-2}$&$8.239\times 10^{-3}$&$4.971\times 10^{-3}$&$2.726\times 10^{-3}$&$1.297\times 10^{-2}$\\
				\midrule
				{Estimate}&$1.575\times 10^{-3}$&$1.575\times 10^{-3}$&$3.006\times 10^{-4}$&$3.006\times 10^{-4}$&$2.065\times 10^{-3}$\\
				\midrule[1pt]
				{Parameter}&{$\mu_{d_1}$}	& {$\mu_{s_1}$}	& {$\mu_{v_1}$}  & {$\dot{\theta}_{t_1}$} &  {$ F_{nt_1} $} \\
				\midrule
				{Initial Guess}&$1.957\times 10^{-2}$&$1.796\times 10^{-2}$&$1.607\times 10^{-2}$&$1.342\times 10^{-2}$&$9.213\times 10^{-3}$\\
				\midrule
				{Estimate}&$3.046\times 10^{-3}$&$3.065\times 10^{-3}$&$3.006\times 10^{-4}$&$3.006\times 10^{-4}$&$2.065\times 10^{-3}$\\
				\bottomrule
			\end{tabular}
		\end{table*}
		\begin{figure*}
			\centering
			\includegraphics[width=14 cm]{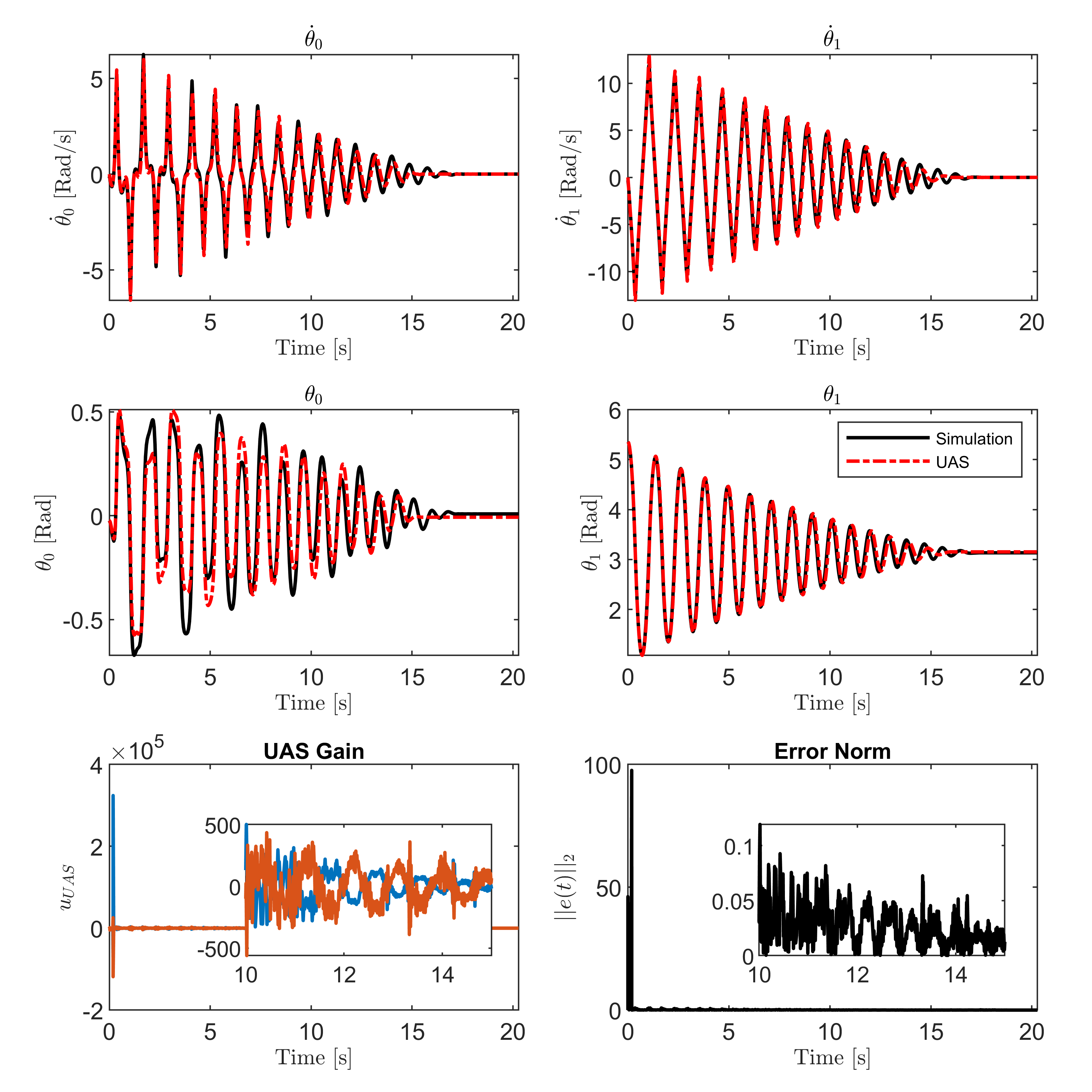}
			\caption{Performance of the UAS observer during the experiment.\label{Exp_UAS}}
		\end{figure*}  
		\begin{figure*}
			\centering
			\includegraphics[width=16 cm]{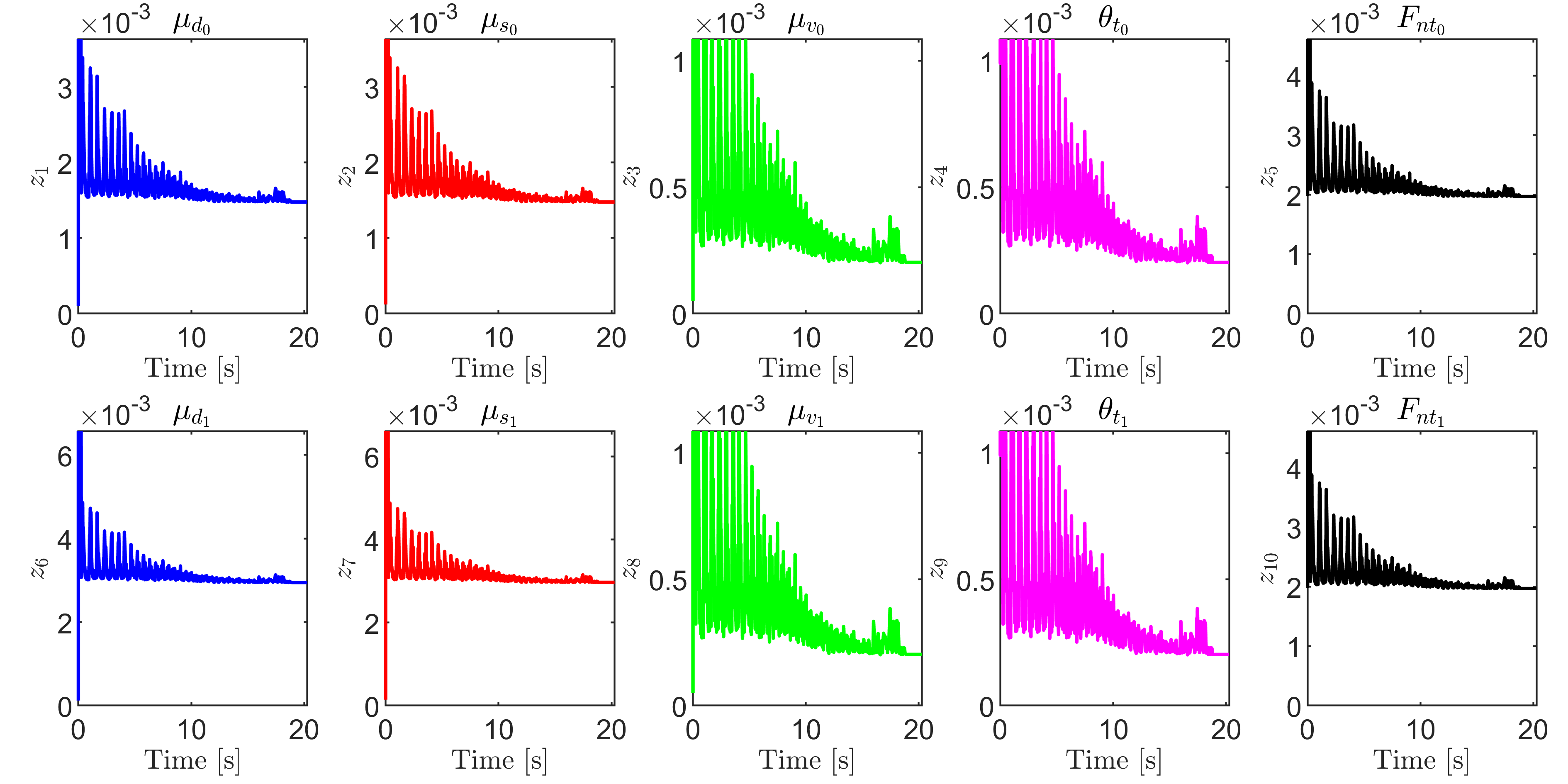}
			\caption{Parameter trajectories during the experiment.\label{Exp_Para}}
		\end{figure*}
		\begin{figure*}
			\centering
			\includegraphics[width=14 cm]{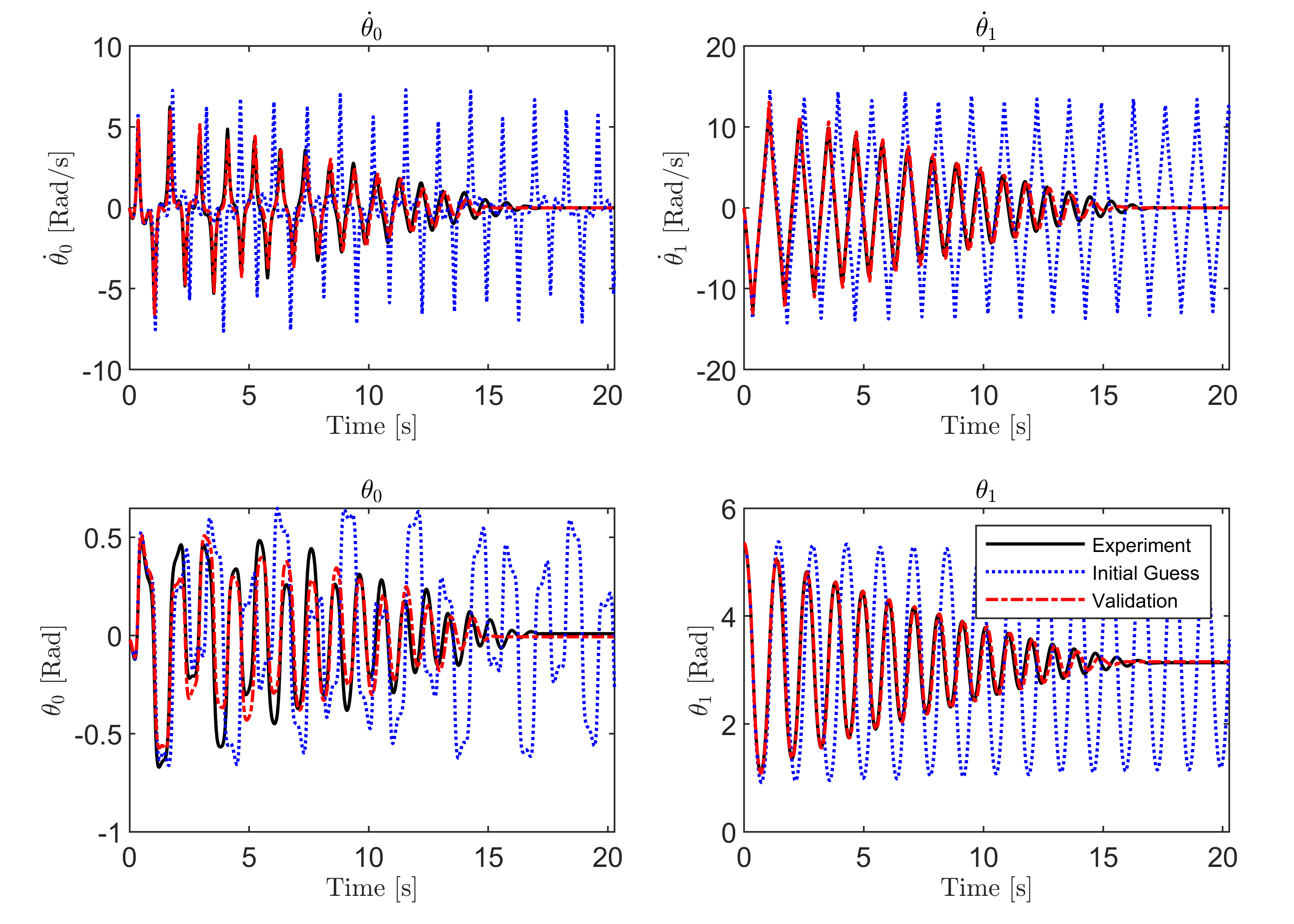}
			\caption{Validation of the identified parameters during the experiment.\label{Exp_Val}}
		\end{figure*}  
		\begin{figure*}
			\centering
			\includegraphics[width=14 cm]{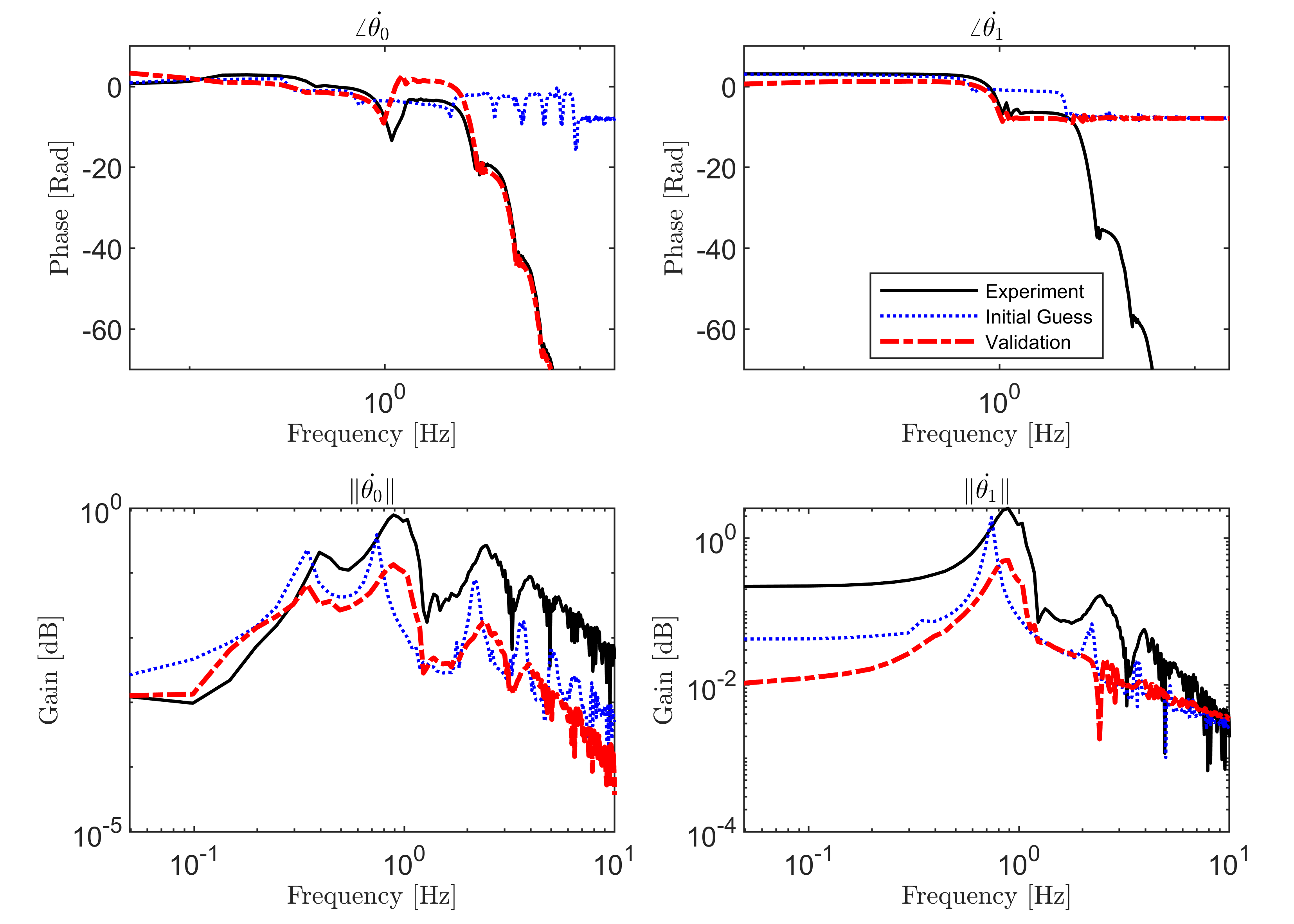}
			\caption{Frequency portrait of the identified parameters from the experiment.\label{Exp_Val_Frq}}
		\end{figure*} 
		
		\subsection{Comparison with conventional methods}
		Here, we compare the proposed approach to a conventional optimization-based one. Specifically, a nonlinear grey-box identification in the MATLAB environment is selected to exploit the known structure of the dynamic model and generate estimates for its parameters. The results for all the cases are shown in Figure~\ref{bar}, where the quality of fit is measured through the coefficient of determination $ R^2 $ metric and is presented in Equation~\ref{NRMSE}.
		
		\begin{equation}\label{NRMSE}
			{fit=100 \times \left(1- \frac{\sum ( \hat{x}-x_{ref})^2 }{\sum(x_{ref} - \bar{x}_{ref} )^2} \right)}
		\end{equation} 
		where $ x_{ref} $ is the reference signal, and $ \hat{x} $ is the estimate of $ x_{ref} $. 
		
		Aditionally, the computation time is computed using
		\begin{equation}\label{T_metric}
			\begin{split}
				&\text{Normalized Computation Time} =\\& \frac{\text{Computation Time}}{\text{Experiment Time} \times \text{Sampling Frequency}}
			\end{split}
		\end{equation}
		
		Here, the same simulation parameters and initial conditions presented in the previous sections are used. Due to the high nonlinearity of the problem, optimization alone is reportedly not always a viable solution to address this problem. Testing shows that in simulation, optimization alone is not capable of reaching acceptable parameter estimates with a goodness of fit of only $ 13\% $ and $ 3\% $ for the Arm Joint and the Pendulum Joint responses, respectively. The proposed algorithm, in contrast, provided quick estimates that provided much better goodness of fit metrics of $ 98\% $ for both responses.  We conclude that the proposed approach is more robust to noisy data than the optimization approach. In the experimental test, the proposed method, while it is able to quickly generate quality estimates, it provides estimates, which can be improved on with a goodness of fit of $ 77\% $ and $ 91\% $ for the responses of the Arm Joint and the Pendulum Joint, respectively. The optimization approach, however, exhibited a low performance in terms of goodness of fit with only $ 3\% $ and $ 24\% $ for the responses of the Arm Joint and the Pendulum Joint, respectively. That result was also achieved at a considerably higher computational cost with optimization taking approximately $ 500$ and $ 1300 $ times more computation time than the UAS approach for the simulation and the experiment, respectively. Table~\ref{table:Comp} summarizes these results.
		\begin{figure}
			\centering
			\includegraphics[width=3in]{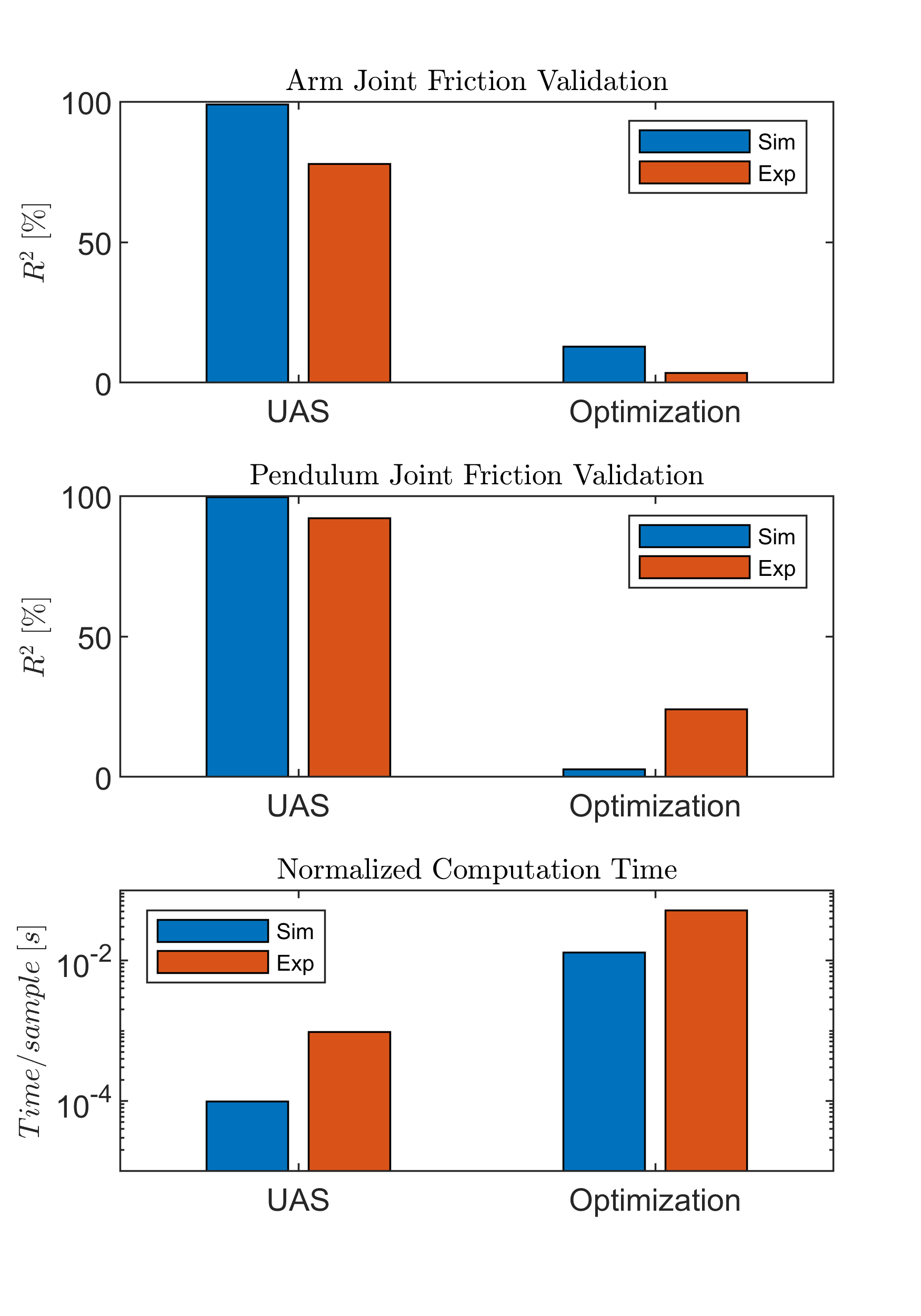}
			\caption{Performance Comparison.\label{bar}}
		\end{figure}  
		
		\begin{table*}
			\centering
			\caption{Computation time and quality of estimates.\\ Computation time is presented along with experiment runtime.}\label{table:Comp}
			\begin{tabular}{llcccc}
				\toprule
				&& \multicolumn{2}{c}{Proposed}	& \multicolumn{2}{c}{Optimization}\\
				\cmidrule(r){3-4} \cmidrule(l){5-6}
				&& Arm Joint&Pendulum Joint& Arm Joint&Pendulum Joint \\
				\midrule
				\multirow{2}{*}{Sim.}&{Test Duration}\quad\quad\,\;[s]			& \multicolumn{2}{c}{$35$}  & \multicolumn{2}{c}{$35$}  	\\
				&{Computation Time}\;[s]			& \multicolumn{2}{c}{$ 3.43$}     & \multicolumn{2}{c}{$ 4504.6 $}	\\
				&		{Goodness of fit}\quad\;\,\,[\%]				& $ 98.97 $		& $ 99.53 $& $12.77$ & $2.80$	\\
				\midrule
				\multirow{2}{*}{Exp.}&{Test Duration}\quad\quad\,\;[s]			& \multicolumn{2}{c}{$20.28$} & \multicolumn{2}{c}{$20.28$}   	\\
				&{Computation Time}\;[s]			& \multicolumn{2}{c}{$ 19.49$ } & \multicolumn{2}{c}{$ 10373.7$ }\\
				&     {Goodness of fit}\quad\;\,\,[\%]					& $77.84$&$91.99$						& $ 3.40 $			& $ 24.11 $	\\
				\bottomrule
			\end{tabular}
		\end{table*}
		\section{Conclusion}
		This paper presented a dynamic analysis of a passive tilted Furuta pendulum. The pendulum was tilted to ensure the existence of a stable equilibrium configuration with the weight of each link as the only external force applied to the system. A complete analytical model was derived where full inertia of the pendulum is taken into account, and a comprehensive friction model was selected form the literature. The problem of parameter estimation of the friction in the setup was addressed. A UAS-based high gain observer was designed to estimate the friction model parameters, and it was tested in simulation and experimentally. The obtained results showed that the developed algorithm is effective at identifying parameters in this class of systems. Nonetheless, it was also concluded that the friction phenomenon is not fully described by the 5-parameter per joint friction model. Simulation results, even with a high level of injected and propagated dynamic noise, were close to the dynamic behavior of the system. Experimental results were qualitatively good when compared to conventional optimization based parameter estimation strategies. A frequency analysis was performed and it showed that the system is highly nonlinear with multiple harmonics detected. Differences between the simulation and the experiments are mainly attributed to the difficulty associated with characterizing friction in the joints. Future work includes the implementation of nonlinear control strategy using the developed model to showcase the performance a more accurate model of the system would permit.
		
		\appendices
		\section{Proof Addendum}\label{APP}
		\begin{remark}\label{rmk}
			Here we show how the inequalities in the mathematical justification are formed.
		\end{remark}
		Take the difference of quadratic forms below where the terms are the same as those in Equation~\eqref{eDynamics6}
		\begin{equation} (e-(\mathcal{B}_1-\mathcal{B}_2)\dot{q})^{T}(e-(\mathcal{B}_1-\mathcal{B}_2)\dot{q})\ge0 \end{equation}
		Expanding the above difference gives
		\begin{equation}
			\begin{split}
				e^Te - e^T(\mathcal{B}_1-\mathcal{B}_2)\dot{q} &- ( (\mathcal{B}_1-\mathcal{B}_2)\dot{q} )^Te \\
				&+ ( (\mathcal{B}_1-\mathcal{B}_2)\dot{q} )^T( (\mathcal{B}_1-\mathcal{B}_2)\dot{q} ) \ge 0 
			\end{split}
		\end{equation}
		Simplifying 
		\begin{equation}
			\begin{split}
				e^Te - e^T(\mathcal{B}_1-\mathcal{B}_2)\dot{q} &- \dot{q}^T(\mathcal{B}_1-\mathcal{B}_2)^Te \\
				&+ \dot{q}^T(\mathcal{B}_1-\mathcal{B}_2)^T(\mathcal{B}_1-\mathcal{B}_2)\dot{q} \ge 0 
			\end{split}
		\end{equation}
		Rearranging 
		\begin{equation}
			\begin{split}
				e^Te + \dot{q}^T(\mathcal{B}_1-\mathcal{B}_2)^T(\mathcal{B}_1-\mathcal{B}_2)\dot{q} \ge 
				2e^T(\mathcal{B}_1-\mathcal{B}_2)\dot{q} 
			\end{split}
		\end{equation}
		where the substitution $  e^T(\mathcal{B}_1-\mathcal{B}_2)\dot{q} +  \dot{q}^T(\mathcal{B}_1-\mathcal{B}_2)^Te = 2e^T(\mathcal{B}_1-\mathcal{B}_2)\dot{q} $ was made following the fact that $ e^T(\mathcal{B}_1-\mathcal{B}_2)\dot{q} $ and $ \dot{q}^T(\mathcal{B}_1-\mathcal{B}_2)^Te $ are scalar quantities implying $ e^T(\mathcal{B}_1-\mathcal{B}_2)\dot{q} =(\dot{q}^T(\mathcal{B}_1-\mathcal{B}_2)^Te)^T $.
		
		The same procedure can be applied on the following quadratic difference form
		\begin{equation} (e-\mathcal{Q})^{T}(e-\mathcal{Q})\ge0 \end{equation}
		which is modified by multiplying $ \mathcal{Q} $ by the $ 2\times2 $ identity matrix $ I_{2\times2} $.
		\begin{equation} (e-I_{2\times2}\mathcal{Q})^{T}(e-I_{2\times2}\mathcal{Q})\ge0 \end{equation}
		which leads to the following after applying the same procedure above
		\begin{equation} e^{T}e + \mathcal{Q}^T\mathcal{Q} \ge 2e^T\mathcal{Q}\end{equation}

		%
		
		%
		
		\vfill

	\end{document}